\documentclass{article}

\usepackage[nonatbib]{neurips_2021}

\usepackage[utf8]{inputenc} 
\usepackage[T1]{fontenc}    
\usepackage{hyperref}       
\usepackage{url}            
\usepackage{booktabs}       
\usepackage{amsfonts}       
\usepackage{nicefrac}       
\usepackage{microtype}      
\usepackage{xcolor}    
\usepackage{graphicx} 

\newcommand{\Appendix}[1]{the full version for}

\renewcommand{\b}{\mathbf{b}}
\renewcommand{\c}{\mathbf{c}}

\renewcommand{\u}{\bm{u}}

\renewcommand{\v}{\bm{v}}

\newcommand{\x}{\bm{x}}
\newcommand{\y}{\bm{y}}
\newcommand{\z}{\bm{z}}

\newcommand{\X}{\bm{X}}

\newcommand{\Z}{\bm{Z}}

\usepackage{bm}
\usepackage{amsfonts}
\usepackage{amsmath}
\usepackage{listings}
\usepackage{xcolor}
\usepackage{hyperref}
\usepackage{amsthm}
\usepackage{float}
\newtheorem{theorem}{Proposition}

\definecolor{codegreen}{rgb}{0,0.6,0}
\definecolor{codegray}{rgb}{0.5,0.5,0.5}
\definecolor{codepurple}{rgb}{0.58,0,0.82}
\definecolor{backcolour}{rgb}{0.95,0.95,0.92}

\lstdefinestyle{mystyle}{
    backgroundcolor=\color{backcolour},   
    commentstyle=\color{codegreen},
    keywordstyle=\color{magenta},
    numberstyle=\tiny\color{codegray},
    stringstyle=\color{codepurple},
    basicstyle=\ttfamily\footnotesize,
    breakatwhitespace=false,         
    breaklines=true,                 
    captionpos=b,                    
    keepspaces=true,                 
    numbers=left,                    
    numbersep=5pt,                  
    showspaces=false,                
    showstringspaces=false,
    showtabs=false,                  
    tabsize=2
}

\title{NFNet: Non-interacting Fermion Network for Efficient Simulation of Large-scale Quantum Systems}

\author{%
}

\begin{document}

\maketitle
\begin{abstract}
    We present NFNet \cite{NFNet-github}, a PyTorch-based framework for polynomial-time simulation of large-scale, continuously controlled quantum systems, supporting parallel matrix computation and auto-differentiation of network parameters. It is based on the non-interacting Fermionic formalism that relates the Matchgates by Valiant \cite{valiant} to a physical analogy of non-interacting Fermions in one dimension as introduced by Terhal and DiVincenzo \cite{Terhal_2002}. Given an input bit string $\x$, NFNet computes the probability $p(\y|\x)=\langle x|U_{\theta}^\dagger \Pi_y U_\theta |x\rangle$ of observing the bit string $\y$, which can be a sub or full-system measurement on the evolved quantum state $U_{\mathbf{\theta}}|x\rangle$, where $\mathbf{\theta}$ is the set of continuous rotation parameters, and the unitary $U_{\mathbf{\theta}}$'s underlying Hamiltonians are not restricted to nearest-neighbor interactions. We first review the mathematical formulation of the Matchgate to Fermionic mapping with additional matrix decomposition derivations, and then show that on top of the pair-wise circuit gates documented in Terhal and DiVincenzo \cite{Terhal_2002}, the Fermionic formalism can also simulate evolutions whose Hamiltonians are sums of arbitrary two-Fermion-mode interactions. We then document the design philosophy of NFNet, its software structure, and demonstrate its usage in tasks such as: \hyperref[subsec:ham]{i)} simulating measurements on continuously evolved quantum states; \hyperref[sec:multicircuit-512]{ii)} modeling 512+ qubit multi-layer continuously-controlled variational quantum circuits; \hyperref[subsec:maxcut]{iii)} finding ground states of classical Hamiltonians such as the weighted-edge Maxcut problem;  \hyperref[subsec:196-mnist]{iv)} training a continuously controlled circuit of $196$ qubits to memorize a binary $14\times 14$ MNIST hand-written digit pattern; and \hyperref[sec:multicircuit-512]{v)} benchmarking Fermionic simulation runtimes of measuring output quantum states of $100$ to $1000$ qubits. As NFNet is both an efficient large-scale quantum simulator, and a quantum-inspired classical computing network structure, many more exciting topics are worth exploring, such as its connection to recurrent neural networks, discrete generative learning and discrete normalizing flow. NFNet source code can be found at \href{https://github.com/BILLYZZ/NFNet}{https://github.com/BILLYZZ/NFNet}.
\end{abstract}

\section{Introduction}
Matchgates were first proposed as a class of two-qubit gates generated by a restricted subset of Pauli's and matchgate circuits were shown to be classically simulable by Valiant \cite{valiant}. Terhal and DiVincenzo \cite{Terhal_2002} further showed that matchgates relate to non-interacting Fermions and extended the simulable interaction to non-nearest-neighbor Fermion modes and conditionally applied circuit gates. Terhal and DiVincenzo \cite{Terhal_2002} systematically derived the simulation of circuit measurements via the computation of Pfaffians, following the decomposition of an anti-symmetric matrix. Other works used various formulations to study the matchgate structures \cite{linearoptics-3} \cite{linearoptics-4} \cite{linearoptics-5}\cite{linearoptics-6}. Moreover, the boundary between simulable matchgates and universal quantum computing is \textit{seemingly} thin and simple. Brod \cite{Brod_2016} provided a good review of these additional resources required to achieve universal computation: the swap gate \cite{linearoptics-6}\cite{linearoptics-8}, two-qubit nondemolition measurements \cite{linearoptics-9}, multi-qubit magic states \cite{linearoptics-10}, parity-preserving two-qubit gates \cite{linearoptics-11}, any connectivity graph that is not a path or a cycle \cite{linearoptics-12}\cite{linearoptics-13}. Furthermore, Brod \cite{Brod_2016} show that matchgates are classically simulable also in the case when the inputs are product states, the measurements are over arbitrarily subset of output qubits with or without adaptive measurements.

We build the Non-interacting Fermion Network (NFNet) \cite{NFNet-github}, a large-scale classical computing pipeline based on the Fermionic formalism introduced by Terhal and DiVincenzo \cite{Terhal_2002}. NFNet allows efficient simulation of continuous time evolution of quantum states, large-scale variational quantum circuits and gradient calculation and optimization of quantum circuit parameters. \textbf{NFNet offers new opportunities} for studying the feasibility of variational quantum algorithms at the hundred-qubit level, simulating entanglement dynamics of large-scale quantum systems, benchmarking quantum machine learning algorithms in high-dimensional data domain, and numerically revealing quantum properties in an otherwise intractable regime.

\textbf{With the goal of simulating as wide range of quantum state evolutions as possible}, we show that the Fermionic formulation also allows simulating continuous evolutions whose Hamiltonians are expressed as sums of two-mode (in the Fermion sense) interactions, which correspond to sums of a restricted set of Pauli strings acting on modes $i, j$ ($i<j$) with $Z$ Pauli's acting on all modes in between. We also complement the proof of Terhal and DiVincenzo \cite{Terhal_2002} with more derivation details especially concerning the decomposition of an anti-symmetric matrix, which is he backbone structure of non-interacting Fermion dynamics. After describing the design philosophy of NFNet, we further provide demonstrations of how to use the library classes and functions provided by NFNet for achieving the aforementioned simulation tasks. The same demonstrations are included in the NFNet release, written as interactive Jupyter Notebook files \cite{NFNet-github}.

The paper is organized as follows:
\begin{itemize}
    \item In Section \ref{sec:math-foundation}, we cover the mathematical foundation of the Fermionic formalism. \begin{itemize}
        \item In Section \ref{subsec:decompose}, we complement the derivations of Terhal and DiVincenzo \cite{Terhal_2002} with rigorous details on the skew-symmetric matrix decomposition that is crucial to the non-interacting Fermion structure.
        \item In Proposition \ref{prop1}, we prove that the strong simulability of NFNet goes beyond the $2$-local quantum circuit setting and extends to dense Hamiltonians written as a sum of arbitrary $(i,j)$-mode Fermionic interactions, corresponding to $k$-local Hamiltonians ($k=|j-i|+1$) each parametrized by $6$ degrees of freedom.
        \item In Proposition \ref{prop2}, we show that the boundary between Fermion-preserving and non-preserving cases is whether the coefficients are pair-wised shared among the strictly $k$-local Hamiltonians. \end{itemize}
        
    \item In Section \ref{connections}, we develop intuitive connections between NFNet and recurrent neural network, power iteration methods, and normalizing flows.
    \item In Section \ref{sec:NFNet-structure}, we explain the mathematical programming structure of the NFNet and show how to initialize an NFNet object, how to change its parameters, and how simulate the measurement probabilities on full and subsets of qubits.
    \item In Section \ref{sec:demos}, we demonstrate the use cases of NFNet with step-by-step coding examples, including: \begin{itemize}
        \item Simulating a continuous Hamiltonian evolution \ref{subsec:ham}
        \item Comparing run times and accuracies between NFNet and exact diagonalization. \ref{subsec:runtime-compare}
        \item Benchmarking NFNet's run times of measuring multi-layer continuous-controlled quantum circuits ranging from $100$ to $1000$ qubits. \ref{sec:multicircuit-512}
        \item Optimizing a quantum Born Machine through auto-differentiation. \ref{subsec:born}
        \item Training a $196$-qubit quantum system to memorize a $14\times 14$ binary MNIST \cite{deng2012mnist} hand-written digit. \ref{subsec:196-mnist}
        \item Approximating the optimal solution to a weighted-edge Maxcut problem with high success rate. \ref{subsec:maxcut}
        \end{itemize}
    \item In Section \ref{sec:conclusion}, we conclude and discuss future extensions.
\end{itemize}

\section{Mathematical Foundation} \label{sec:math-foundation}
\subsection{Preliminary: Hamiltonians constructed quadratically in Fermion operators}
Given ani-commutation relations of Fermionic creators and annihilators: 
\begin{equation}
    \left\{a_{i}, a_{j}\right\} \equiv a_{i} a_{j}+a_{j} a_{i}=0, \quad\left\{a_{i}^{\dagger}, a_{j}^{\dagger}\right\}=0, \quad\left\{a_{i}, a_{j}^{\dagger}\right\}=\delta_{i j} I,
\end{equation}
the Majorana Fermions can be viewed as a change of basis of the origininal forms $a_i$ and $a_i ^\dagger$:
\begin{equation}
    c_{2 i}=a_{i}+a_{i}^{\dagger}, \quad c_{2 i+1}=-i\left(a_{i}-a_{i}^{\dagger}\right),
\end{equation}
where the anti-commutation relations are: $\left\{c_{k}, c_{l}\right\}=2 \delta_{k l} I$, where $c_l c_l=I$ and $c_l = c_l^\dagger.$

The crucial structure that allows a Hamiltonian $H$ to be re-expressed as non-interacting Fermion dynamics is the quadratic composition \cite{Terhal_2002}, where the $k\neq l$ constraint is to avoid the creation of scalar constants as $c_k c_k = I$.:
\begin{equation} \label{H}
    H=\frac{i}{4} \sum_{k \neq l=1}^{2n} \alpha_{k l} c_{k} c_{l} = \frac{i}{4} \sum_{k < l}^{2n} (\alpha_{k l}-\alpha_{l k}) c_{k} c_{l}
\end{equation}
Because $H$ is Hermitian:
\begin{equation}
    H^\dagger =\frac{-i}{4} \sum_{k < l}^{2n} (\overline{\alpha_{k l}}-\overline{\alpha_{l k}}) c_{l} c_{k}=\frac{i}{4} \sum_{k < l}^{2n} (\overline{\alpha_{k l}}-\overline{\alpha_{l k}}) c_{k} c_{l},
\end{equation}
which requires that 
\begin{equation}
\begin{split}
    \overline{\alpha_{k l}}-\overline{\alpha_{l k}}=\alpha_{k l}-\alpha_{l k}\\
    \implies \text{img}(\alpha_{k l}) = \text{img}(\alpha_{l k}).
\end{split}
\end{equation}
Due to the anticommuting relationship, it suffices to re-express the most general case $\alpha_{kl} = a+bi$ and $\alpha_{lk} = c+bi$ as:
\begin{equation}
    \alpha_{kl} = d+bi; \; \alpha_{lk} = -d+bi.
\end{equation}
This is because $a$ and $c$ only enter the evaluation of H \ref{H} by $(\alpha_{kl}-\alpha_{lk})= a-c$, and any difference between two real numbers can be expressed by another real number $2d$.

In addition, because the imaginary parts of $\alpha_{kl}$ and $\alpha_{lk}$ cancel out in the evaluation of H \ref{H}, the set $\{\alpha_{k l}\}$ is hence generalized by any real skew-symmetric matrix $A$ where $A_{kl}=\alpha_{kl}$, $\alpha_{k l}=-\alpha_{l k}$.

The above process follows that documented in \cite{Terhal_2002}. It is exactly through a skew-symmetric matrix decomposition that the classical simulability arises, corresponding to non-interacting Fermionic dynamics. We now complement the matrix decomposition derivation in \cite{Terhal_2002} with more details.

\subsection{Skew-symmetric matrix decomposition} \label{subsec:decompose}
First note that from the spectral theorem, for a real skew-symmetric matrix the nonzero eigenvalues are all pure imaginary and are pair-wise complex conjugates: 
${\displaystyle \lambda _{1}i,-\lambda _{1}i,\lambda _{2}i,-\lambda _{2}i,\ldots }$ where each of the $\lambda_k$ is real.

The ``canonical block diagonal form'' (Terhal 2000, Valiant, etc) relies on the fact that for a real number $\lambda$:
\begin{equation}
    \frac{1}{2}\left(\begin{array}{rr}
-i & i \\
1 & 1
\end{array}\right)\left(\begin{array}{ll}
\lambda i & \\
& -\lambda i
\end{array}\right)\left(\begin{array}{ll}
i & 1 \\
-i & 1
\end{array}\right)=\left(\begin{array}{cc}
0 & \lambda \\
-\lambda & 0
\end{array}\right).
\end{equation}
If we define an orthogonal matrix $O=\frac{1}{\sqrt{2}}\left(\begin{array}{rr}
-i & i \\
1 & 1
\end{array}\right)$, the above becomes:
\begin{equation}
    O\left(\begin{array}{ll}
\lambda i & \\
& -\lambda i
\end{array}\right)O^\dagger = \left(\begin{array}{cc}
0 & \lambda \\
-\lambda & 0
\end{array}\right).
\end{equation}
It is easy to see that this works in a block-diagonal setting:
\begin{equation}
\begin{split}
     &Block(O)Block_{k=1}^b\left[\left(\begin{array}{ll}
\lambda_k i & \\
& -\lambda_k i
\end{array}\right)\right]Block(O^\dagger) = Block_{k=1}^b \left[\left(\begin{array}{cc}
0 & \lambda_k \\
-\lambda_k & 0
\end{array}\right)\right]\\
&=\left(\begin{array}{ccccc}
0 & \lambda_{1} & & & \\
-\lambda_{1} & 0 & & & \\
& & \ddots & & \\
& & & 0 & \lambda_{b} \\
& & & & \\
& & & -\lambda_{b} & 0
\end{array}\right).
\end{split}
\end{equation}

Putting it together, for a real skew-symmetric matrix $A$ (writing $Block(O)$ as $\mathbf{O}$):
\begin{equation}
\begin{split}
    &A = U\Sigma U^\dagger=U\left(\begin{array}{ccccc}
\lambda_{1}i & 0 & & & \\
0 & -\lambda_{1}i & & & \\
& & \ddots & & \\
& & & \lambda_{b}i & 0 \\
& & & & \\
& & & 0 & -\lambda_{b}i
\end{array}\right)U^\dagger\\
&=U \mathbf{O}^\dagger \mathbf{O} \left(\begin{array}{ccccc}
\lambda_{1}i & 0 & & & \\
0 & -\lambda_{1}i & & & \\
& & \ddots & & \\
& & & \lambda_{b}i & 0 \\
& & & & \\
& & & 0 & -\lambda_{b}i
\end{array}\right) \mathbf{O}^\dagger \mathbf{O} U^\dagger
= \underbrace{U \mathbf{O}^\dagger}_{W^\dagger} \left(\begin{array}{ccccc}
0 & \lambda_{1} & & & \\
-\lambda_{1} & 0 & & & \\
& & \ddots & & \\
& & & 0 & \lambda_{b} \\
& & & & \\
& & & -\lambda_{b} & 0
\end{array}\right) \underbrace{\mathbf{O} U^\dagger}_{W},
\end{split}
\end{equation}
where note that $W^\dagger W = I$. 

We further note that W is a real matrix, because for the real matrix A where its eigenvalues come in conjugate pairs, the eigenvectors are also conjugate pairs:
\begin{equation}
\begin{split}
      A\v = \lambda \v
\end{split}
\end{equation}

taking conjugate on both sides
\begin{equation}
      \overline{A\v}=\overline{A} \overline{\v}= A \overline{\v}=\overline{\lambda} \overline{\v}
\end{equation}
therefore, $\v$ and $\overline{\v}$ are eigenvectors corresponding to eigenvalues $\lambda$ and $\overline{\lambda}$.  We can thus write the orthogonal basis matrix U of A as:
\begin{equation}
    U=\begin{bmatrix}
      \u_1+i\v_1, \u_1-i\v_1, \u_2+i\v_2, \u_2-i\v_2...,\u_k+i\v_k, \u_k+i\v_k
    \end{bmatrix},
\end{equation}
where $\u$ and $\v$'s are real-valued vectors. The transformation $\mathbf{O}^\dagger$ gives:
\begin{equation}
    W^\dagger = U\mathbf{O}^\dagger = \frac{2}{\sqrt{2}}[-\v_1, \u_1, -\v_2, \u_2,...]
\end{equation}
meaning that  W is indeed a real matrix. Thus $W^\dagger=W^T$.

The important step now is the \textbf{canonical Fermionic block-diagonal transformation}, using substitution $\b = W \c$, i.e., $\c = W^\dagger \b$: 
\begin{equation}
\begin{split}
    H&=\frac{i}{4} \sum_{k \neq l=1}^{2n} \alpha_{k l} c_{k} c_{l}=\frac{i}{4}\begin{bmatrix}c_0, c_1, c_2, c_3,...,c_{2n-2}, c_{2n-1}\end{bmatrix} A \begin{bmatrix}
        c_0\\ c_1\\ \vdots \\ c_{2n-2} \\ c_{2n-1}
    \end{bmatrix}\\
    &=\frac{i}{4} \underbrace{\c^\dagger}_{(W^\dagger \b )^\dagger} A \underbrace{\c}_{W^\dagger \b}=\frac{i}{4} (W^\dagger \b )^\dagger A W^\dagger \b=\frac{i}{4} \b^\dagger W \underbrace{W^\dagger \left(\begin{array}{ccccc}
0 & \lambda_{1} & & & \\
-\lambda_{1} & 0 & & & \\
& & \ddots & & \\
& & & 0 & \lambda_{b} \\
& & & & \\
& & & -\lambda_{b} & 0
\end{array}\right) W}_{A} W^\dagger \b\\
&= \frac{i}{4}\b^\dagger \left(\begin{array}{ccccc}
0 & \lambda_{1} & & & \\
-\lambda_{1} & 0 & & & \\
& & \ddots & & \\
& & & 0 & \lambda_{b} \\
& & & & \\
& & & -\lambda_{b} & 0
\end{array}\right) \b.
\end{split}
\end{equation}

By making the center matrix into a block form, we successfully decomposed the Fermionic system into a sum of \textbf{local contributions}. There are no interactions between different sites (\textbf{non-interacting}). Note that the transformed elements inside $\b=\left(\begin{array}{l}
b_1 \\
b_2 \\
\vdots \\
b_{2n-1} \\
b_{2n}
\end{array}\right)$ have the same anti-commutation relations as $c_j$'s, because:
\begin{equation}
\begin{split}
    \left\{b_{k}, b_{l}\right\} &\equiv b_k b_l+b_l b_k = W_{(k,\cdot)}\c W_{(l,\cdot)}\c + W_{(l,\cdot)}\c W_{(k,\cdot)}\c \\
    &= W_{(k,\cdot)}\c \c^T W_{(l,\cdot)}^T + W_{(l,\cdot)}\c \c^T W_{(k,\cdot)}^T=2 \delta_{k l} I,
\end{split}
\end{equation}
and $b_k ^\dagger = b_k$.

Therefore, we can express H as

\begin{equation}
    H = \frac{i}{2}\sum_{j}^n \lambda_j b_j 'b_j '',
\end{equation}
where we re-express $\b$ as
\begin{equation}
    \b=\left(\begin{array}{l}
b_{0}^{\prime} \\
b_{0}^{\prime \prime} \\
\vdots \\
b_{n-1}^{\prime} \\
b_{n-1}^{\prime \prime}
\end{array}\right),
\end{equation}
where each two elements represent one Fermionic mode.

\textbf{Now, the key building block is} the calculation of 
\begin{equation}
\begin{split}
    U c_{i} U^{\dagger}&= \exp \left(-\frac{1}{2} \sum_{m} \epsilon_{m} b_{m}^{\prime} b_{m}^{\prime \prime}\right)\sum_{j=0}^{n-1} \left(W_{2 j, i} b_{j}^{\prime}+W_{2 j+1, i} b_{j}^{\prime \prime}\right) \exp \left(\frac{1}{2} \sum_{m} \epsilon_{m} b_{m}^{\prime} b_{m}^{\prime \prime}\right)\\
    &=\sum_{j} \exp \left(-\frac{1}{2} \sum_{m} \epsilon_{m} b_{m}^{\prime} b_{m}^{\prime \prime}\right)\left(W_{2 j, i} b_{j}^{\prime}+W_{2 j+1, i} b_{j}^{\prime \prime}\right) \exp \left(\frac{1}{2} \sum_{m} \epsilon_{m} b_{m}^{\prime} b_{m}^{\prime \prime}\right).
\end{split}
\end{equation}
because for different indices, $b_m' b_m''$ and $b_k' b_k''$ commute, we can write the matrix exponential as the product of a series of exponentials:
\begin{equation}
\begin{split}
     \sum_{j} &\exp \left(-\frac{1}{2} \sum_{m} \epsilon_{m} b_{m}^{\prime} b_{m}^{\prime \prime}\right)\left(W_{2 j, i} b_{j}^{\prime}+W_{2 j+1, i} b_{j}^{\prime \prime}\right) \exp \left(\frac{1}{2} \sum_{m} \epsilon_{m} b_{m}^{\prime} b_{m}^{\prime \prime}\right)\\
     &=  \sum_{j} \prod_m \exp \left(-\frac{1}{2} \epsilon_{m} b_{m}^{\prime} b_{m}^{\prime \prime}\right)\left(W_{2 j, i} b_{j}^{\prime}+W_{2 j+1, i} b_{j}^{\prime \prime}\right) \prod_n \exp \left(\frac{1}{2} \epsilon_{n} b_{n}^{\prime} b_{n}^{\prime \prime}\right).
\end{split}
\end{equation}
in the first group product of matrix exponentials, exponentials involving $b_{m}^{\prime} b_{m}^{\prime \prime}$, $m\neq j$ can freely move/commute pass the middle term:
\begin{equation}
\begin{split}
     &\sum_{j} \prod_m \exp \left(-\frac{1}{2} \epsilon_{m} b_{m}^{\prime} b_{m}^{\prime \prime}\right)\left(W_{2 j, i} b_{j}^{\prime}+W_{2 j+1, i} b_{j}^{\prime \prime}\right) \prod_n \exp \left(\frac{1}{2} \epsilon_{n} b_{n}^{\prime} b_{n}^{\prime \prime}\right)\\
     &=\sum_{j} \exp \left(-\frac{1}{2} \epsilon_{j} b_{j}^{\prime} b_{j}^{\prime \prime}\right)\left(W_{2 j, i} b_{j}^{\prime}+W_{2 j+1, i} b_{j}^{\prime \prime}\right)\exp \left(\frac{1}{2} \epsilon_{j} b_{j}^{\prime} b_{j}^{\prime \prime}\right)  \prod_{m\neq j} \exp \left(-\frac{1}{2} \epsilon_{m} b_{m}^{\prime} b_{m}^{\prime \prime}\right) \prod_{n\neq j} \exp \left(\frac{1}{2} \epsilon_{n} b_{n}^{\prime} b_{n}^{\prime \prime}\right)\\
     &=\sum_{j} \exp \left(-\frac{1}{2} \epsilon_{j} b_{j}^{\prime} b_{j}^{\prime \prime}\right)\left(W_{2 j, i} b_{j}^{\prime}+W_{2 j+1, i} b_{j}^{\prime \prime}\right)\exp \left(\frac{1}{2} \epsilon_{j} b_{j}^{\prime} b_{j}^{\prime \prime}\right) \exp \left(-\frac{1}{2} \sum_{m\neq j} \epsilon_{m} b_{m}^{\prime} b_{m}^{\prime \prime}\right)  \exp \left(\frac{1}{2} \sum_{n\neq j} \epsilon_{n} b_{n}^{\prime} b_{n}^{\prime \prime}\right)\\
     &=\sum_{j} \exp \left(-\frac{1}{2} \epsilon_{j} b_{j}^{\prime} b_{j}^{\prime \prime}\right)\left(W_{2 j, i} b_{j}^{\prime}+W_{2 j+1, i} b_{j}^{\prime \prime}\right)\exp \left(\frac{1}{2} \epsilon_{j} b_{j}^{\prime} b_{j}^{\prime \prime}\right).
\end{split}
\end{equation}
Note that we make the substitution 
\begin{equation}
    \begin{split}
        \begin{aligned}
\exp \left(\frac{1}{2} \epsilon_{j} b_{j}^{\prime} b_{j}^{\prime \prime}\right) &=\sum_{k=0}^{\infty} \frac{\left(\epsilon_{j} / 2\right)^{k}}{k !}\left(b_{j}^{\prime} b_{j}^{\prime \prime}\right)^{k}=\sum_{k=0}^{\infty} \frac{\left(\epsilon_{j} / 2\right)^{2 k}}{2 k !}(-1)^{k}+\sum_{k=0}^{\infty} \frac{\left(\epsilon_{j} / 2\right)^{2 k+1}}{(2 k+1) !}(-1)^{k} b_{j}^{\prime} b_{j}^{\prime \prime} \\
&=\cos \left(\epsilon_{j} / 2\right)+b_{j}^{\prime} b_{j}^{\prime \prime} \sin \left(\epsilon_{j} / 2\right)
\end{aligned}
    \end{split}
\end{equation}
to make the above expression into
\begin{equation}
    \begin{split}
       & \sum_{j} \exp \left(-\frac{1}{2} \epsilon_{j} b_{j}^{\prime} b_{j}^{\prime \prime}\right)\left(W_{2 j, i} b_{j}^{\prime}+W_{2 j+1, i} b_{j}^{\prime \prime}\right)\exp \left(\frac{1}{2} \epsilon_{j} b_{j}^{\prime} b_{j}^{\prime \prime}\right)\\
       &=\sum_{j} (\cos \left(\epsilon_{j} / 2\right)-b_{j}^{\prime} b_{j}^{\prime \prime} \sin \left(\epsilon_{j} / 2\right) ) \left(W_{2 j, i} b_{j}^{\prime}+W_{2 j+1, i} b_{j}^{\prime \prime}\right) (\cos \left(\epsilon_{j} / 2\right)+b_{j}^{\prime} b_{j}^{\prime \prime} \sin \left(\epsilon_{j} / 2\right))\\
       &=\sum_{j} \cos \epsilon_{j} W_{2 j, i} b_{j}^{\prime}+\sin \epsilon_{i} W_{2 j, i} b_{j}^{\prime \prime}+ \cos \epsilon_{j} W_{2 j+1, i} b_{j}^{\prime \prime}-\sin \epsilon_{j} W_{2 j+1, i} b_{j}^{\prime}\\
       &=(W^T)_{i,\cdot}\left(\begin{array}{cccc}
\cos \epsilon_{0}& \sin \epsilon_{0} & & \\
-\sin \epsilon_{0} & \cos \epsilon_{0} & & \\
& & \ddots & \\
& & & \\
& & \cos \epsilon_{n-1} & \sin \epsilon_{n-1} \\
& & -\sin \epsilon_{n-1} & \cos \epsilon_{n-1}
\end{array}\right) \b\\
&= (W^T)_{i,\cdot}\left(\begin{array}{cccc}
\cos \epsilon_{0}& \sin \epsilon_{0} & & \\
-\sin \epsilon_{0} & \cos \epsilon_{0} & & \\
& & \ddots & \\
& & & \\
& & \cos \epsilon_{n-1} & \sin \epsilon_{n-1} \\
& & -\sin \epsilon_{n-1} & \cos \epsilon_{n-1}
\end{array}\right) W\c.
    \end{split}
\end{equation}
if we stack the matrices $ U c_{i} U^{\dagger}$ in a column block fashion:
\begin{equation} \label{svd}
    \left(\begin{array}{c}
         U c_{0} U^{\dagger}  \\
         \vdots\\
         U c_{2n-1} U^{\dagger}
    \end{array}\right)=W^T \left(\begin{array}{cccc}
\cos \epsilon_{0}& \sin \epsilon_{0} & & \\
-\sin \epsilon_{0} & \cos \epsilon_{0} & & \\
& & \ddots & \\
& & & \\
& & \cos \epsilon_{n-1} & \sin \epsilon_{n-1} \\
& & -\sin \epsilon_{n-1} & \cos \epsilon_{n-1}
\end{array}\right) W \c=R \c.
\end{equation}
In other words, we can write:
\begin{equation} \label{conjugate}
    U c_{i} U^{\dagger}=\sum_{j=0}^{2 {n}-1} R_{i j} c_{j}=R_{i,\cdot} \c.
\end{equation}

\subsection{Multi-gate conjugate composition} 
We express in more details based on \cite{Terhal_2002}, that in a multi-gate circuit setting, the conjugate operation  \ref{conjugate} with a total circuit unitary $U$ can be decomposed into the product of a series of matrix multiplication of each gate's $R$ matrix.

For example, for a circuit with two gates, $U_1$ and $U_2$, the conjugate is:
\begin{equation}
    \begin{split}
        &U_2 U_1 c_{i} U_1^{\dagger} U_2^\dagger = U_2\left( \sum_{j=0}^{2 {n}-1} R^{(1)}_{i j} c_{j}\right) U_2^\dagger\\
        &=\sum_{j=0}^{2 {n}-1} U_2 R^{(1)}_{i j} c_{j} U_2^\dagger=\sum_{j=0}^{2 {n}-1} R^{(1)}_{i j} U_2 c_{j} U_2^\dagger\\
        &=\sum_{j=0}^{2 {n}-1} R^{(1)}_{i j} R^{(2)}_{j,\cdot} \c= R^{(1)}_{i,\cdot} R^{(2)}\c,
    \end{split}
\end{equation}
where $R^{(1)}$ and $R^{(2)}$ correspond to the circuit Hamiltonians $H_1$ and $H_2$ respectively. Therefore, for arbitrary amounts (m) of gates, we have:
\begin{equation}
    U_m ... U_1 c_{i} U_1^{\dagger}... U_m^\dagger=R^{(1)}_{i,\cdot} R^{(2)} ... R^{(m)}\c
\end{equation}

\subsection{Measurement computation} \label{sec:measurement_computation}
We briefly summarize the derivation of \cite{Terhal_2002}. For a general input bit string $\x$, its corresponding quantum state is expressed via Fermionic creators operated on the $|\mathbf{0}\rangle$ state:
\begin{equation}
|x\rangle=a_{p_1}^{\dagger} \ldots a_{p_l}^{\dagger}|\mathbf{0}\rangle.
\end{equation}

Written in terms of Majorana operators, this is equivalently:
\begin{equation}
c_{2 p_1} \ldots c_{2 p_l}|{\mathbf{0}}\rangle,
\end{equation}
where $1 \leq p_1 < p_2<\cdots \leq N$ are positions where the qubit is $1$. 

For any given quantum state $|\psi\rangle = \sum_i^{2^N} \alpha_i |x_i\rangle$, the matrix denoted by $a^{\dagger}_i a_i$ will select all amplitudes $\alpha_i$ whose corresponding bases $|x_i\rangle$ has the $i$-th bit equals to $1$; similarly, $a_i a^{\dagger}_i$ will select amplitudes whose bases has the $i$-th bit equals to $0$. Therefore, given an initial single basis state $|x\rangle$ that is evolved under a unitary $U$, the probability of observing a bit string $\y$ (a sub-string of length $k\leq N$) is (for example, measuring the $j_1$-th bit to be $1$, $j_2$-th bit to be $0$, and $j_k$-th bit to be $1$ ):
\begin{equation}
    p(y|x) = \langle x|U^\dagger (a^\dagger_{j_1} a_{j_1})(a_{j_2} a^\dagger_{j_2})...(a^\dagger_{j_k} a_{j_k})  U|x\rangle.
\end{equation}

Via conjugation, we have:
\begin{equation}
p\left(y | x\right)=\langle\mathbf{0}\left|c_{2 p_l} \ldots c_{2 p_1}\right| U^{\dagger} a_{j_1}^{\dagger} U U^{\dagger} a_{j_1} U \ldots U^{\dagger} a_{j_k}^{\dagger} U U^{\dagger} a_{j_k} U c_{2 p_1} \ldots c_{2 p_l} \mid \mathbf{0}\rangle,
\end{equation}
where
\begin{equation}
\begin{split}
    U^{\dagger} a_i U=\frac{1}{2} U^{\dagger}\left(c_{2 i}+i c_{2 i+1}\right) U&=\frac{1}{2} \sum_j\left({R_{2 i, i}^T}+{i R_{2 i+1, j}^T}\right) c_j=\sum_j T_{i j} c_j,\\
    U^{\dagger} a_i^{\dagger} U&=\sum_j T_{i j}^* c_j.
\end{split}
\end{equation}
Thus
\begin{equation}
\begin{aligned}
p\left(y | x\right)= & \sum_{m_1, n_1, \ldots, m_k, n_k} T_{j_1, m_1} T_{j_1, n_1}^* \ldots T_{j_k, n_k}^* T_{j_k, m_k} \\
& \left\langle\mathbf{0}\left|c_{2 p_l} \ldots c_{2 p_1} c_{m_1} c_{n_y} \ldots c_{n_k} c_{m_k} c_{2 p_1} \ldots c_{2 p_l}\right| \mathbf{0}\right\rangle.
\end{aligned}
\end{equation}
Then, through Wick's theorem [2], the computation of $p\left(y | x\right)$ equates to the evaluation of the Pfaffian of matrix $M$, where $M$ is a function of matrix elements of $T$ (refer to the appendices of Terhal and DiVincenzo \cite{Terhal_2002} for more details on the remapping tables).

\subsection{Hamiltonian composition of multi-pair interactions} \label{subsec:proposition}
Valiant \cite{valiant} and Terhal et al \cite{Terhal_2002} mainly considered the simulation of quantum circuit in the sense that each gate is a unitary $e^{-iH}$ where $H$ is a Hamiltonian concerning interactions only involving two Fermion modes $i,j$. We herein show that as long as the skew-symmetric matrix structure is kept, the efficient non-interacting Fermion simulation also holds for the case when $H$ is a sum of pair-wise Fermion interactions.
\begin{theorem} \label{prop1}
For any Hamiltonian written as a sum of quadratic Majorana terms
\begin{equation} \label{H}
H=\frac{i}{4} \sum_{k \neq l=1}^{2n} \alpha_{k l} c_{k} c_{l} = \frac{i}{4} \sum_{k < l}^{2n} (\alpha_{k l}-\alpha_{l k}) c_{k} c_{l},
\end{equation}
where $\{\alpha_{kl}\}$ form a matrix $A$ such that $A_{kl}=\alpha_{kl}$ which is a real skew-symmetric matrix, $H$ can always be written as the sum of pair-wise Fermion interactions, 
\begin{equation}
H = \frac{i}{4} \sum_{i < j} H^{(i,j)},
\end{equation}
where each $H^{(i,j)}$ is controlled by at most $6$ degrees of freedom and is $k$-local where $k=j-i+1$.
\end{theorem}

\begin{proof}
We prove this via re-expressing the $H=\frac{i}{4} \sum_{k \neq l=1}^{2n} \alpha_{k l} c_{k} c_{l}=\frac{i}{4}\c^\dagger A \c=\frac{i}{4}\c^\top A \c$ in a block form:
\begin{equation}
\begin{split}
    H &=\frac{i}{4}\c^\top A \c=\frac{i}{4} (\underbrace{c_0, c_1}_{\xi^\top_0}, \underbrace{c_2, c_3}_{\xi^\top_1},\cdots, \underbrace{c_{2(n-1)}, c_{2(n-1)+1}}_{\xi^\top_{n-1}}) A \begin{pmatrix}
        c_0\\ c_1 \\ \vdots \\ c_{2(n-1)} \\ c_{2(n-1)+1}
    \end{pmatrix}\\
    &=\frac{i}{4}(\xi^\top_0,\xi^\top_1,\cdots, \xi^\top_{n-1})\begin{pmatrix} B_0 &B_{0, 1} &B_{0, 2} &\dots &B_{0, n-1}\\
    B_{1, 0} &B_1 &B_{1, 2} &\dots &B_{1, n-1}\\
    B_{2,0} &B_{2, 1} &B_2 &\dots &B_{2, n-1}\\
    \vdots  &\vdots   &\vdots   &\ddots & \vdots\\
    B_{n-1, 0} &B_{n-1, 1} &B_{n-1, 2} &\dots &B_{n-1}
    \end{pmatrix}\begin{pmatrix}
        \xi_0 \\ \xi_1 \\ \vdots \\ \xi_{n-1}
    \end{pmatrix},
\end{split}
\end{equation}
where because $A$ is skew-symmetric, for $i<j$:
\begin{equation}
    B_i = B_{i,i} = \begin{pmatrix}
        0 &f^{(i)}\\
        -f^{(i)} &0
    \end{pmatrix}, \; B_{i,j}=\begin{pmatrix}
        a^{(i,j)}, &b^{(i,j)}\\
        c^{(i,j)}, &d^{(i,j)}
    \end{pmatrix} \; B_{j,i}=\begin{pmatrix}
        -a^{(i,j)}, &-c^{(i,j)}\\
        -b^{(i,j)}, &-d^{(i,j)}
    \end{pmatrix}. 
\end{equation}
Therefore, 
\begin{equation}
\begin{split}
H = \frac{i}{4} \sum_{i,j} \xi^\top_i B_{i,j} \xi_j=\frac{i}{4} \sum_{(i,j)=e\in E, i < j} (\xi^\top_i, \xi^\top_j) &\begin{pmatrix}B^{e}_i &B_{i,j}\\
B_{j,i} & B^{e}_j
\end{pmatrix} \begin{pmatrix}
    \xi_i \\ \xi_j
\end{pmatrix} =\frac{i}{4} \sum_{i < j}H^{(i, j)},\\
\sum_{(k, j)=e\in E, k<j} B^e_k + \sum_{(j, k)=e\in E, j<k} B^e_k   &= B_k= B_{k,k} = \begin{pmatrix}
        0 &f^{(k)}\\
        -f^{(k)} &0
    \end{pmatrix}
\end{split}
\end{equation}
which means that $H$ is a sum of Hamiltonians involving any pair of Fermion modes $i<j$. \end{proof}

\subsection{Fermion to Pauli Mapping}
For each term $\frac{i}{4}H^{(i,j)}$, the exact mapping from the Fermionic notation and Pauli strings directly follows:
\begin{equation}
\begin{split}
    \frac{i}{4}H^{(i,j)} &= \frac{i}{4} (\xi^\top_i, \xi^\top_j) \begin{pmatrix}B^{e}_i &B_{i,j}\\
B_{j,i} & B^{e}_j
\end{pmatrix} \begin{pmatrix}
    \xi_i \\ \xi_j
\end{pmatrix}=\frac{i}{4}(c_{2i}, c_{2i+1}, c_{2j}, c_{2j+1})\begin{pmatrix} 
&0 &e &a &b\\
&-e &0 &c &d\\
&-a &-c &0 &f\\
&-b &-d &-f &0
\end{pmatrix}\begin{pmatrix}c_{2i}\\ c_{2i+1}\\ c_{2j}\\ c_{2j+1}\end{pmatrix}\\
&=\frac{i}{4}(a_{i}+a_{i}^\dagger, -i(a_{i}-a_{i}^\dagger), a_{j}+a_{j}^\dagger, -i(a_{j}-a_{j}^\dagger))\begin{pmatrix} 
&0 &e &a &b\\
&-e &0 &c &d\\
&-a &-c &0 &f\\
&-b &-d &-f &0
\end{pmatrix}\begin{pmatrix}a_{i}+a_{i}^\dagger\\ -i(a_{i}-a_{i}^\dagger)\\ a_{j}+a_{j}^\dagger\\ -i(a_{j}-a_{j}^\dagger)\end{pmatrix}\\
&=\frac{-e}{2} Z_i+\frac{-f}{2} Z_j + \frac{-c}{2} X_i \left(\prod_{k=i+1}^{j-1} Z_k\right) X_j + \frac{b}{2} Y_i \left(\prod_{k=i+1}^{j-1} Z_k\right) Y_j +\\
&\frac{-d}{2} X_i \left(\prod_{k=i+1}^{j-1} Z_k\right) Y_j + \frac{a}{2}Y_i \left(\prod_{k=i+1}^{j-1} Z_k\right) X_j,
\end{split}
\end{equation}
where we directly apply 
\begin{align} \label{JW}
&a_j^{\dagger}:=\left(\prod_{k=1}^{j-1} Z_k\right)\left(\frac{X_j-i Y_j}{2}\right) \\
&a_j:=\left(\prod_{k=1}^{j-1} Z_k\right)\left(\frac{X_j+i Y_j}{2}\right).
\end{align}

\subsection{Special Case: Fermion-preserving Evolution}
We here showcase in more detail the Fermion-preserving evolution as documented in \cite{Terhal_2002}. The cases with Fermion-preserving Hamiltonian evolution are the easiest to compute. A Fermion-preserving Hamiltonian expressed in terms of Fermion creators and annihilators is:
\begin{equation}
H_g=b_{i i} a_i^{\dagger} a_i+b_{j j} a_j^{\dagger} a_j+b_{i j} a_i^{\dagger} a_j+b_{i j}^* a_j^{\dagger} a_i.
\end{equation}
Because $H_g$ is Hermitian, $\{b_{ij}\}_{(i,j)}=\b$ forms a $2$ by $2$ Hermitian matrix.

\textbf{We hereby emphasize that $i,j$ do not have to be nearest neighbors.} The following derivation maps the Fermion-preserving $H_g$ to tensor products of Pauli operators for any arbitrary pair of $i,j$ where $i\neq j$.
\begin{theorem} \label{prop2}
    The Fermion-preserving Hamiltonian has restricted degrees of freedom ($4$ as opposed to $6$) and incurs two-Fermion-mode interactions between modes $(i, j)$, corresponding to the sum of $k$-local Hamiltonians where $k=j-i+1$. The boundary between Fermion-preserving and non-preserving cases is whether the coefficients are pair-wised shared among the strictly $k$-local Hamiltonians.
\end{theorem}

\begin{proof}
We start from (using Eq. \ref{JW}):
\begin{equation}
\begin{aligned}
&b_{i i} a_i^{+} a_i=\left(\prod_{k=1}^{i-1} Z_k\right)\left(\frac{X_i-i Y_i}{2}\right)\left(\prod_{k=1}^{i-1} Z_k\right)\left(\frac{X_i+i Y_i}{2}\right)=b_{i i} \left(\frac{I}{2}-\frac{Z_i}{2}\right)\\
&b_{j j} a_j^{+} a_j=\left(\prod_{k=1}^{j-1} Z_k\right)\left(\frac{X_j-i Y_j}{2}\right)\left(\prod_{k=1}^{j-1} Z_k\right)\left(\frac{X_j+i Y_j}{2}\right)=b_{j j} \left(\frac{I}{2}-\frac{Z_j}{2}\right).
\end{aligned}
\end{equation}

In addition, we expand to get (assuming $i < j$):
\begin{equation}
\begin{aligned}
&a_i^\dagger a_j=\left(\prod_{k=1}^{i-1} Z_k\right)\left(\frac{X_i-i Y_i}{2}\right) \left(\prod_{k=1}^{j-1} Z_k\right)\left(\frac{X_j+i Y_j}{2}\right)\\
&=\frac{(X_i-i Y_i)Z_i}{2}\left(\prod_{k=i+1}^{j-1} Z_k \right) \frac{(X_j+i Y_j)}{2}\\
&=\frac{1}{4}\left(X_i Z_i \left(\prod_{k=i+1}^{j-1} Z_k \right) X_j + X_i Z_i \left(\prod_{k=i+1}^{j-1} Z_k\right) iY_j - iY_i Z_i \left(\prod_{k=i+1}^{j-1} Z_k\right) X_j-iY_i Z_i \left(\prod_{k=i+1}^{j-1} Z_k\right) iY_j  \right),\\
&a_j^{\dagger} a_i = \left(\prod_{k=1}^{j-1} Z_k\right)\left(\frac{X_j-i Y_j}{2}\right) \left(\prod_{k=1}^{i-1} Z_k\right)\left(\frac{X_i+i Y_i}{2}\right)\\
&=\frac{1}{4}\left(Z_i X_i \left(\prod_{k=i+1}^{j-1} Z_k \right) X_j - Z_i X_i \left(\prod_{k=i+1}^{j-1} Z_k\right) iY_j + iZ_i Y_i \left(\prod_{k=i+1}^{j-1} Z_k\right) X_j-iZ_i Y_i \left(\prod_{k=i+1}^{j-1} Z_k\right) iY_j  \right).
\end{aligned}
\end{equation}

Then, by expressing $b_{ij}=c+di$ and $b^{*}_{ij}=c-di$, we get:
\begin{equation}
    \begin{aligned}
        b_{ij} a_i^{\dagger} a_j + b_{ij}^{*} a_j^{\dagger} a_i &=\frac{1}{2}c\left(Y_i \left(\prod_{k=i+1}^{j-1} Z_k\right)Y_j\right)+\frac{1}{2}c\left(X_i \left(\prod_{k=i+1}^{j-1} Z_k\right)X_j\right)\\
        &+\frac{1}{2}d\left(Y_i \left(\prod_{k=i+1}^{j-1} Z_k\right) X_j\right)-\frac{1}{2}d\left(X_i \left(\prod_{k=i+1}^{j-1} Z_k\right) Y_j \right).
    \end{aligned}
\end{equation}
Therefore, the total $H_g$ applied on any pair of $i < j$ is (where we write $b_{ii}=a$ and $b_{jj}=b$):
\begin{equation}\label{H_preserve}
    \begin{aligned}
        H_g^{i,j} &= -\frac{1}{2} \left(aZ_i + bZ_j\right)+\frac{c}{2} \left(Y_i \left(\prod_{k=i+1}^{j-1} Z_k\right) Y_j + X_i \left(\prod_{k=i+1}^{j-1} Z_k\right) X_j \right)\\
        &+\frac{d}{2}\left(Y_i \left(\prod_{k=i+1}^{j-1} Z_k\right) X_j - X_i \left(\prod_{k=i+1}^{j-1} Z_k\right) Y_j \right),
    \end{aligned}
\end{equation}
where $a, b, c, d$ are real continuous coefficients.
\end{proof}

\subsection{Limitation: Non-nearest-neighbor Fermion Modes Dynamics}
We emphasize here that the efficient simulation of non-interacting Fermions comes with a cost, which limits universal computation. The problem arises with the non-interacting Fermion formulation is the additional single $\Z$ rotations introduced, corresponding to the phase adjustment of creating or destroying a Fermion. To see this effect, consider a single Hamiltonian evolution with $H_g^{i,j}$ applied to the first and last qubits of single basis state $|s\rangle$ (where $s_i$ is the $i$-th bit of string s). For example:
\begin{equation}
\begin{split}
    e^{-i H_g^{1,N}}& |s_1\rangle |s_{2,3,...,N-1}\rangle |s_N \rangle = \exp \left(-i X_i \left(\prod_{k=i+1}^{j-1} Z_k\right) X_j \right) |s_1\rangle |s_{2,3,...,N-1}\rangle |s_N \rangle\\
    &=e^{ \left(-i \lambda \X_{1} \X_N \right) } |s_1\rangle |s_{2,3,...,N-1}\rangle |s_N \rangle,
\end{split}
\end{equation}
where $\lambda = (-1)^{\sum_{j=2}^{N-1} s_j}$ is the eigenvalue of the matrix $\prod_{k=i+1}^{j-1} Z_k$ with eigenvector $|s_{2,3,...,N-1}\rangle$. In other words, the $\prod_{k=i+1}^{j-1} Z_k$ \textbf{picks up an additional phase for the Hamiltonian} used to evolve the state, which cannot be directly offset via a global phase adjustment on the total state at the end.

This phase change introduces an interesting dynamic, where if there are odd number of spins in the $|s_{1,...,N-1}\rangle$ string, \textbf{the evolution applies a inverse evolution} on the $i$-th and $N$-th qubit (node) as opposed to the even number case. Thus decreases the expressiveness as the evolution is conditioned on other bits.

\section{Connections to Recurrent Neural Network Unrolling, Power Iteration, and Normalizing Flows} \label{connections}
The mathematical operations inherent to the NFNet computations have key connections to the Recurrent Neural Network (RNN). The ``recurrent'' element resides in the fact that, to compute the singular values of the real skew-symmetric matrix $A$ in the calculation of conjugate operators $U c_i U^\dagger$ (Eq. \ref{svd}) (replicated below), one has to numerically solve the compact singular value decomposition (SVD) via power iteration or its variants:
\begin{equation} \label{svd}
    \left(\begin{array}{c}
         U c_{0} U^{\dagger}  \\
         \vdots\\
         U c_{2n-1} U^{\dagger}
    \end{array}\right)=W^T \left(\begin{array}{cccc}
\cos \epsilon_{0}& \sin \epsilon_{0} & & \\
-\sin \epsilon_{0} & \cos \epsilon_{0} & & \\
& & \ddots & \\
& & & \\
& & \cos \epsilon_{n-1} & \sin \epsilon_{n-1} \\
& & -\sin \epsilon_{n-1} & \cos \epsilon_{n-1}
\end{array}\right) W \c=R \c,
\end{equation}
where $W=\boldsymbol{O}U^\dagger$. The singular values $\{\epsilon_i\}$ and singular vector matrix $U$ of $A$ are computed via power iteration $(\{\epsilon_i\}, U)=\mathcal{P}(A)$. For example, to compute the SVD of a matrix $A \in \mathbb{R}^{2n \times 2n}$. We start by computing the first singular value $\sigma_1$ and left and right singular vectors $u_1$ and $v_1$ of $A$:
\begin{enumerate}
    \item Generate $x_0$ such that $x_0(i) \sim \mathcal{N}(0,1)$.
    \item $s$: number of iteractions
    \item for $i$ in $[1, \ldots, s]$: $\quad x_i \leftarrow A^T A x_{i-1}$ end for
    \item $v_1 \leftarrow x_i /\left\|x_i\right\|$
    \item $\sigma_1 \leftarrow\left\|A v_1\right\|$
    \item $u_1 \leftarrow A v_1 / \sigma_1$
    \item return $\left(\sigma_1, u_1, v_1\right)$.
\end{enumerate}
The $s$, the number of iterations is lower-bounded by a given error $\varepsilon$ via $s\geq \log (4 \log (2 n / \delta) / \varepsilon \delta) / 2 \lambda$ in order to get $\varepsilon$ precision with probability at least $1-\delta$, and $\min _{i<j} \; \log \left(\sigma_i / \sigma_j\right) \geq \lambda$. Step 3. corresponds to a recurrent connection. For the rest of the singular vectors and singular values, simply update $A$ via $A\leftarrow A-\sigma_1 u_1 v_1^T=\sum_{i=2}^n \sigma_i u_i v_i^T$, \textbf{which corresponds to a residual connection after the recurrent layer.}

Furthermore, the inherent probabilistic nature of quantum states naturally implies that the densities values of both the initial observation probability $p_{0}(\y|\x)$ and the final observation probability $p_{\theta}(\y|\x)$ are normalized (fixing $\x$) on a discrete support. Therefore, \textbf{the important implication is that one can use a quantum system/circuit to represent a latent distribution} (which is conjectured to be hard to simulate via classical neural networks \cite{boixo_supremacy}), and \textbf{NFNet is an efficient classical framework that naturally enables transformations of discrete distributions} by simulating the quantum behavior. However, the missing piece is how to find or learn the mapping $f_{\phi}(\y): \y\rightarrow \z$ of discrete data to discrete latent variables such that the log likelihood computed using the circuit's output density $p_{\theta}(\z|\x)$, i.e., $\mathcal{LL}=\sum_{\z} \log(p_{\theta}(\z|\x))=\sum_{\y} \log(p_{\theta}(f_{\phi}(\y)|\x))$ is maximized, where the calculation of $\mathcal{LL}$ is done by sampling from the circuit in a real quantum setting.

\section{NFNet Structure} \label{sec:NFNet-structure}
\subsection{Formulation}
Mathematically, NFNet is a layer-wise continuously parametrized network. Each layer models a unitary corresponding to a Hamiltonian $H_l$ which is generalized by:
\begin{equation}
    H_l = \frac{i}{4} \sum_{i < j}H^{(i, j)}(\boldsymbol{\alpha}_l, \boldsymbol{\beta}_l) = \frac{i}{4} \sum_{i < j} H_{1}^{(i,j)}(\alpha_1, \beta_1) + H_{2}^{(i,j)}(\alpha_2, \beta_2) + H_{3}^{(i,j)}(\alpha_3, \beta_3).
\end{equation}
Therefore, each layer simulates unitary evolution $e^{-iH_l}$ parameterized by $6\frac{N(N-1)}{2}$ real rotation parameters.

A network consists of $L$ total layers simulates the total unitary evolution written as:
\begin{equation}
    U_{\theta} = e^{-iH_L}\dots e^{-iH_2}e^{-iH_1},
\end{equation}
with a total of $6L$ parameters.

The NFNet computes $p_{\theta}(y|x)$, where $\x$ is a $N$-bit string and $\y$ is a $k$-bit string, $k\leq N$, and $\theta$ is the total set of network parameters. Quantum-mechanically, $p_{\theta}(y|x)$ is the probability of measurement outcome $\y$ given a basis state $|x\rangle$. Therefore, the network directly outputs the discrete density $p_{\theta}(y|x)$ satisfying:
\begin{equation}
    \sum_{\boldsymbol{y}\in \{0,1\}^k} p_{\theta}(y|x) = 1.
\end{equation}

We here emphasize that NFNet does not directly compute $U_\theta$ by the exponential-time exact diagonalization, but rather computes $p_{\theta}(y|x)=Pf(M(\boldsymbol{\alpha}, \boldsymbol{\beta}))$, where $M(\boldsymbol{\alpha}, \boldsymbol{\beta})\in \mathbb{R}^{2N\times 2N}$ is a matrix parametrized by all rotation parameters in the network:
\begin{equation}
    M(\boldsymbol{\alpha}, \boldsymbol{\beta})=f(T(\boldsymbol{\alpha}, \boldsymbol{\beta}))=f(T(R(\boldsymbol{\alpha}_1, \boldsymbol{\beta}_1)R(\boldsymbol{\alpha}_2, \boldsymbol{\beta}_2)\dots R(\boldsymbol{\alpha}_L, \boldsymbol{\beta}_L))),
\end{equation}
where $f(\cdot)$ is a re-mapping procedure that re-orders elements in the $T$ matrix according the tables in \cite{Terhal_2002}. $T$ and $R$ matrices are as defined in Section \ref{sec:measurement_computation}.

Because the width of $M(\boldsymbol{\alpha},\boldsymbol{\beta})$ scales linearly with $N$, the evaluation of  and the Pfaffian calculation is polynomial in the size of $M(\boldsymbol{\alpha},\boldsymbol{\beta})$, whose computation (matrix multiplication and polynomial-time re-sorting) is also polynomial in $N$, the total computation of $p_{\theta}(y|x)$ is thus also polynomial in $N$.

\subsection{Initialize an NFNet}
An NFNet ``Network'' class inherits a PyTorch nn.Module, which contains continuously differentiable parameters. To initialize an NFNet ``network'' class, only the pre-defined number of qubits $N$ and the Fermionic connectivities are required. In the following example we consider a circuit structure where each gate $l$ has an underlying two-Fermion-mode Hamiltonian $H^{(i,j)}_l$. $N>0$ is a real integer, and ``conn\textunderscore list''is a nested list, inside which each pair (connectivity of each gate) is a specific connection between two Fermion modes $i,j$. The length of ``conn\textunderscore list'' is thus the number of total layers $L$.
\begin{lstlisting}[language=Python, caption=Initialize NFNet.Network]
from NFNet import Network, get_nn_pairs, binary_basis, unpacknbits, initialize_sparse

N = 10 # 10 qubits
conn_list = [[1, 6], [2, 3], [1, 4]] # We consider the Hamiltonian interaction between qubit 1 and 6, 2 and 3, and 1 and 4 (zero indexed).
L = len(conn_list)
circuit = Network(conn_list, N) # Create an evolution simulation, calling the NFNet.Network class.

\end{lstlisting}

\subsection{Change NFNet parameters}
Upon initialization, the NFNet ``Network'' class automatically select random rotation parameters. The user can also manually set these rotation parameters via the ``Network.manual\textunderscore set\textunderscore params()'' function. 
\begin{lstlisting}[language=Python, caption=Print parameters]
from NFNet import Network, get_nn_pairs, binary_basis, unpacknbits, initialize_sparse

N = 10 # 10 qubits
conn_list = [[1, 6], [2, 3], [1, 4]] # We consider the Hamiltonian interaction between qubit 1 and 6, 2 and 3, and 1 and 4 (zero indexed).
L = len(conn_list)
circuit = Network(conn_list, N) # Create an evolution simulation, calling the NFNet.Network class.

# Print the parameters of the circuit evolution:
for p in circuit.parameters():
    print(p)
\end{lstlisting}
which prints the following:
\begin{lstlisting}[language=Python]
Parameter containing:
tensor([0.1341], requires_grad=True)
Parameter containing:
tensor([2.1245], requires_grad=True)
Parameter containing:
tensor([1.9884], requires_grad=True)
Parameter containing:
tensor([2.4108], requires_grad=True)
param Parameter containing:
tensor([0.5980], requires_grad=True)
param Parameter containing:
tensor([2.2515], requires_grad=True)
\end{lstlisting}

To manually set the circuit's rotation parameters, call ``Network.manual\textunderscore set\textunderscore params()'', which takes a PyToch tensor of shape $(L, 4)$ or $(L, 6)$ (recall that each $(i,j)$ Fermion mode interaction Hamiltonian is parametrized by $4$ real numbers in the Fermion-preserving case and $6$ in the general case).
\begin{lstlisting}[language=Python, caption=Manually set parameters]
from NFNet import Network, get_nn_pairs, binary_basis, unpacknbits, initialize_sparse

N = 10 # 10 qubits
conn_list = [[1, 6], [2, 3], [1, 4]] # We consider the Hamiltonian interaction between qubit 1 and 6, 2 and 3, and 1 and 4 (zero indexed).
L = len(conn_list)
circuit = Network(conn_list, N) # Create an evolution simulation, calling the NFNet.Network class.

# Randomly sample the real a, b, c, d parameters, where a = b_ii, b = b_jj, c+di = b_ij
params = torch.tensor(math.pi) * torch.rand((L, 4)) # Randomly initialize parameters
circuit.manual_set_params(params) # Set the evolution parameters
\end{lstlisting}

\subsection{Compute full-system measurement probability}
NFNet supports measurements on all qubits or a subset of qubits. We separate these two functionalities into two class methods. For the case of a full measurement, an ``NFNet.Network'' object takes two inputs, $\y$ and $\x$, where $\x$ an $N$-bit string representing is the single basis state $|x\rangle$ that is input to the circuit, and $\y$ is an $N$-bit string corresponding to a specific measurement result. NFNet.Network class computes $p_{\theta}(y|x)=|\langle y|U_{\theta}|x\rangle|^2$ via the ``NFNet.Network.forward(y, x)'' function.

Because NFNet supports batch processing, $\x$ and $\y$ both have shape (batch\textunderscore size, N). In the following example, we use a batch size of $1$ circuit whose nearest-neighbor connectivity scheme is repeated $5$ times:
\begin{lstlisting}[language=Python, caption=Compute full-measurement probability]
from NFNet import Network, get_nn_pairs, binary_basis, unpacknbits, initialize_sparse

N = 10 # 10 qubits
conn_list = [[np.min(p), np.max(p)] for p in get_nn_pairs((N,))]*5 # Fully connected (zero-indexed Fermion modes)
L = len(conn_list)
circuit = Network(conn_list, N) # Create an evolution simulation, calling the NFNet.Network class.
x_batch = torch.tensor([[1,0]*(N//2)]) #1010...10
y_batch = torch.tensor([[1]*(N//2)+[0]*(N//2)]) #11...100...0
prob = torch.abs(circuit.forward(y_batch, x_batch))**2
print(prob) # p(y|x)
\end{lstlisting}
\begin{lstlisting}
    tensor([0.0021], grad_fn=<PowBackward0>)
\end{lstlisting}

\subsection{Compute sub-system measurement probability}
NFNet also supports measurements on subsystems. An ``NFNet.Network'' object takes two inputs, $\y$ and $\x$, where $\x$ an $N$-bit string representing is the single basis state $|x\rangle$ that is input to the circuit, and $\y$ is an $k$-bit string ($k\leq N$) corresponding to a specific measurement result \textbf{on a subset of qubits}. NFNet.Network class computes $p_{\theta}(y|x)=\langle x|U_{\theta}^{\dagger}\Pi_y U_{\theta}|x\rangle$ via the ``NFNet.Network.forward\textunderscore partial\textunderscore observation(y, x)'' function.

Because NFNet supports batch processing, $\x$ and $\y$ both have shape (batch\textunderscore size, N). On top of $\x$ and $\y$, a third input, an $N$-bit measurement mask string $\boldsymbol{m}$ is required, where we only measure the $i$-th qubit if $\boldsymbol{m}_{i}=1$. In the following example, we use a batch size of $1$ circuit whose pair-wise fully connected scheme is repeated $5$ times:
\begin{lstlisting}[language=Python, caption=Compute sub-measurement probability]
from NFNet import Network, get_nn_pairs, binary_basis, unpacknbits, initialize_sparse

N = 10 # 10 qubits
conn_list = [[np.min(p), np.max(p)] for p in get_nn_pairs((N,))]*5 # Fully connected (zero-indexed Fermion modes)
L = len(conn_list)
circuit = Network(conn_list, N) # Create an evolution simulation, calling the NFNet.Network class.
x_batch = torch.tensor([[1,0]*(N//2)]) # Input state N-bit string 1010101010
y_batch = torch.tensor([[1,1,1,0,0]]) # Measurement outcome N//2 bit string 11100
mask_batch = torch.tensor([[1]*(N//2)+[0]*(N//2)])# Mask 1111100000 only measure the first 5 qubits
prob = circuit.forward_partial_observation(y_batch, x_batch, mask_batch).real
print(prob) # p(y|x)
\end{lstlisting}
\begin{lstlisting}
    tensor([0.0174], grad_fn=<SelectBackward0>)
\end{lstlisting}

\section{NFNet Use Case Tutorials} \label{sec:demos}
The following use cases are also included in the NFNet release in forms of interactive Jupyter Notebooks.
\subsection{Simulate a continuous Hamiltonian evolution} \label{subsec:ham}
In this example, we first show how to use the free Fermion formalism to map a continuous Hamiltonian (in terms of Pauli matrices) to the Fermionic creators and annihilators. Then we simulate the time evolution using the free Fermion simulation. 

We measure on all qubits in the Z basis after the evolution and compare the measurement probabilities with the exact diagonalization simulation result.



In this example we look at a Fermion-preserving Hamiltonian on two fermion modes $i$ and $j$, written as (assuming $i<j$):

\begin{equation}
\mathcal{H}_{Fermi}=b_{i i} a_i^{\dagger} a_i+b_{j j} a_j^{\dagger} a_j+b_{i j} a_i^{\dagger} a_j+b_{i j}^* a_j^{\dagger} a_i.
\end{equation}

Because $H$ is hermitian, $a_{ii}, a_{jj}$ are real. We express $b_{ij}$ as $c+di$, where $c$ and $d$ are real continuous numbers.

In the most general case, $H$ corresponds to a Hamiltonian (in Pauli matrix form) which describes the interaction between qubits $i$ and $j$:

\begin{equation}
    \begin{aligned}
        \mathcal{H}_{Pauli} &=b_{i i} a_i^{\dagger} a_i+b_{j j} a_j^{\dagger} a_j+(c+di) a_i^{\dagger} a_j+(c-di) a_j^{\dagger} a_i \\ &= -\frac{1}{2} \left(b_{ii} Z_i + b_{jj} Z_j\right)+\frac{c}{2} \left(Y_i \left(\prod_{k=i+1}^{j-1} Z_k\right) Y_j + X_i \left(\prod_{k=i+1}^{j-1} Z_k\right) X_j \right)\\
        &+\frac{d}{2}\left(Y_i \left(\prod_{k=i+1}^{j-1} Z_k\right) X_j - X_i \left(\prod_{k=i+1}^{j-1} Z_k\right) Y_j \right),
    \end{aligned}
\end{equation}

In this example we first simulate the Z-basis measurement probabilities on the final state $|\psi_{f} \rangle = e^{-i\mathcal{H} t} |\psi_0 \rangle$ after the time evolution of $e^{-i\mathcal{H} t}$ on an input product state $|\psi_0 \rangle$. This simulation is computed in polynomial time of the number of qubits $N$.

We then compare the measurement probabilities with the exact diagonalization calculation, keeping track of the full state. This scales exponential time of the number of qubits $N$.

\begin{lstlisting}[language=Python, caption=Demo 1: Simulate a continuous Hamiltonian evolution]
# Load NFNet modules:
from NFNet import Network, get_nn_pairs, binary_basis, unpacknbits, initialize_sparse
import numpy as np
import matplotlib.pyplot as plt
import scipy
import scipy.linalg
import time
import torch # PyTorch is required for this demo
import math

# First we create a simulation object, from the Fermion "Network" class.
N = 10 # 10 qubits
conn_list = [[1, 6]] # We consider the Hamiltonian interaction between qubit 1 and 6 (zero indexed)
L = len(conn_list)
evolution = Network(conn_list, N) # Create an evolution simulation, calling my PyFerm module
x_input = torch.tensor([[1,0]*(N//2)]) # The input state is the 1010101... single basis state

# Randomly sample the real a, b, c, d parameters, where a = b_ii, b = b_jj, c+di = b_ij
params_abcd = torch.tensor(math.pi) * torch.rand((L, 4)) # Randomly initialize parameters
t = 1.5 # For how long do we evolve the state
evolution.manual_set_params(t*params_abcd) # Set the evolution parameters

# Now calculate the probability of P(y|x_input) for all possible y bitstrings:
basis_m_n = torch.tensor(binary_basis(geometry=(N,))) # shape is number of y bitstrings by 2^N
probs = np.zeros(2**N)
# The evolution Network object can handle batch processing of p(y_batch|x_input_batch)
y_batch = basis_m_n
x_batch = x_input.repeat(y_batch.shape[0], 1) # shape is number of y bitstrings by 2^N

# This is a mask tensor that tells the network which qubits to measure at the end
# For example, [1111111111] indicates that we meausure on all 10 qubits in the end
# [1111100000] indicates that we measure on the first 5 qubits in the end
measure_mask_batch = (torch.tensor([[1]*N])).repeat(y_batch.shape[0], 1) # shape is number of y bitstrings by N

ts = time.time()
# The returned measurement probabilities are torch tensors, don't forget to detach and convert to numpy
probs_batch = evolution.forward_partial_observation(y_batch, x_batch, measure_mask_batch).detach().numpy()
tf = time.time()
\end{lstlisting}

If we compare to exact diagonalization:
\begin{lstlisting}[language=Python]
import qiskit # Use qiskit to conveniently convert bit strings to statevectors

x_string = '10'*int(N/2)+'1' if N%2==1 else '10'*int(N/2)
init_state_numpy = qiskit.quantum_info.Statevector.from_label(x_string).data
print('input x string', x_string)
ts = time.time()

# Initialize the exact e^{-iHt} evolution matrix
exp_iH_exact = np.eye(2**N)

conn = conn_list[0]
a, b, c, d = params_abcd.detach().numpy()[0]
H_exact = initialize_sparse(N, conn[0], conn[1], a, b, c, d)
exp_iH_exact = (scipy.linalg.expm(-t*1.0j*H_exact))@exp_iH_exact

state_exact = np.matmul(exp_iH_exact, init_state_numpy[:,None])
probs_exact = (np.abs(state_exact)**2).squeeze()
tf = time.time()

# Calculate the sum of absolute differences in density values
tv = np.abs(probs_batch-probs_exact).sum()
print('Total variation between Fermion and Exact simulations', tv)
# Plot the probabilities from exact diagonalization to fermion simulation. Yay they match!
plt.figure(figsize=(8,5))
plt.plot(probs_batch, '^-')
plt.plot(probs_exact, 'x-')
plt.xlabel('y (in base 10)')
plt.ylabel('prob')
plt.legend(['Fermion probs', 'Exact Diagalization probs'])
\end{lstlisting}
\begin{figure}[h!]
    \centering
    \includegraphics[scale=0.25]{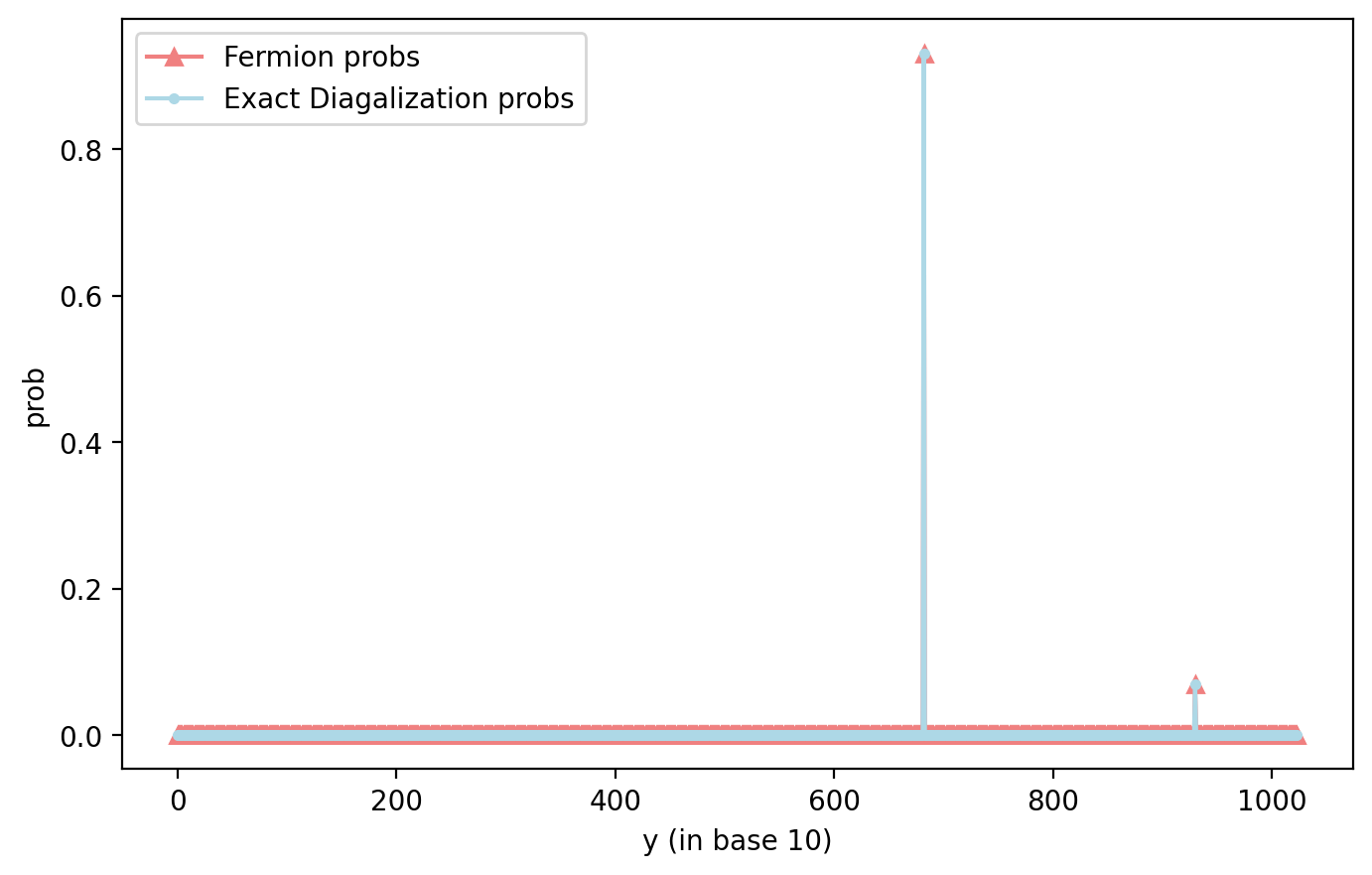}
    \caption{The probabilities of measurements as calculated by Fermionic and exact diagonalization for a Fermion-preserving continuous-time Hamiltonian evolution.}
    \label{fig:my_label}
\end{figure}

\subsection{Compare runtimes and simulation accuracy} \label{subsec:runtime-compare}
In this example, we compare the runtimes of the free Fermion simulation versus the exact diagonalization simulation, on the same continuous Hamiltonian evolution as in Demo 1:
\begin{equation}
    \begin{aligned}
        \mathcal{H}_{Pauli} &=b_{i i} a_i^{\dagger} a_i+b_{j j} a_j^{\dagger} a_j+(c+di) a_i^{\dagger} a_j+(c-di) a_j^{\dagger} a_i \\ &= -\frac{1}{2} \left(b_{ii} Z_i + b_{jj} Z_j\right)+\frac{c}{2} \left(Y_i \left(\prod_{k=i+1}^{j-1} Z_k\right) Y_j + X_i \left(\prod_{k=i+1}^{j-1} Z_k\right) X_j \right)\\
        &+\frac{d}{2}\left(Y_i \left(\prod_{k=i+1}^{j-1} Z_k\right) X_j - X_i \left(\prod_{k=i+1}^{j-1} Z_k\right) Y_j \right).
    \end{aligned}
\end{equation}

We conduct the runtime comparison with the following benchmarking procedure (N is the number of qubits in the system):

Set x\textunderscore input state as a simple product state (in this case we use $1010...10$).

For N in [2, 4, 6, 8, 10]:
\begin{enumerate}
    \item Evolve the initial state by $e^{-i\mathcal{H}t}$ to a final state $|\psi_{f}\rangle$ (for the case of exact diagonalization).
    \item Calculating the probability of observing bit string 1010...10, which is $p(y=101010|x=101010)=|\langle1010...10| e^{-i\mathcal{H}t}|1010...10\rangle|^2$
    \item Record and compare the time to get $p(y=101010|x=101010)$ for Fermion simulation and for exact diagonalization
\end{enumerate}

\begin{lstlisting}[language=Python]
# Load NFNet modules:
from Utils_torch_version import Network, get_nn_pairs, binary_basis, unpacknbits, initialize_sparse
import numpy as np
import matplotlib.pyplot as plt
import scipy.linalg
import time
import torch # PyTorch is required for this demo
import math
import qiskit # Qiskit is required for this demo

Ns = [2, 4, 6, 8, 10]
reps = 5
avg_time_at_N_fermion, avg_time_at_N_exact = [], []
std_time_at_N_fermion, std_time_at_N_exact = [], []
avg_abs_diff_prob_at_N, std_abs_diff_prob_at_N = [], []

for N in Ns:
    print('start N: ', N)
    conn_list = [[0, N-1]] 
    times_fermion = []
    times_exact = []
    diffs = [] # to collect the differences between two simulation methods
    for rep in range(reps):
        L = len(conn_list) # Number of layers
        # initiliaze the evolution
        evolution = Network(conn_list, N)
        
        # The 1010101... basis state as input
        x = torch.tensor([[1,0]*int(N/2)]) if N%2==0 else torch.tensor([[1,0]*int(N/2)+[1]])  
        
        # Create statevector for ED
        x_string = '10'*int(N/2)+'1' if N%2==1 else '10'*int(N/2)
        init_state_numpy = qiskit.quantum_info.Statevector.from_label(x_string).data

        # Fix this, the parameters are defined differently now for the pytorch implementation
        params_m = torch.tensor(math.pi) * torch.rand((L, 4))
        
        t = 1.5 # Time to evolve the state for. Can set to other arbitrary values
        evolution.manual_set_params(t*params_m) # load parameterse to the evolution

        #basis_m_n = torch.tensor(binary_basis(geometry=(N,)))

        probs = torch.zeros(2**(N), dtype=torch.cfloat)
        ts = time.time()
        
        y = x # Calculate the probability of observing 101010... Can choose other bit strings too
        # A mask that tells the network which qubits to measure in the end
        # Eg. 11111 means measuring all qubits, 11000, means measurin the first two qubits
        sub_mask = torch.tensor([[1]*N])
        
        # Since we aren't doing optimization, simply record the numerical value, detach gradients.
        probs_fermion = \
            evolution.forward_partial_observation(y, x, sub_mask).detach().numpy()
        
        # Record the Fermion simulation's runtime
        tf = time.time()
        times_fermion.append(tf - ts)

        ts = time.time() # Reset timer
        
        # Calculate the evolution matrix e^{-iHt}
        exp_iH_exact = np.eye(2**N)
        for l in range(L):
            conn = conn_list[l]
            a, b, c, d = params_m.detach().numpy()[l]
            H_exact = initialize_sparse(N, conn[0], conn[1], a, b, c, d)
            exp_iH_exact = (scipy.linalg.expm(-t*1.0j*H_exact))@exp_iH_exact # 1.0j or -1.0j?
        
        # Exact final state via ED
        state_exact = np.matmul(exp_iH_exact, init_state_numpy[:,None])
        
        ind = int(''.join([str(i) for i in y.detach().numpy().squeeze()]), 2) 
        # The measurement probabilities calculated via ED:
        probs_exact = (np.abs(state_exact[ind])**2).squeeze()
        tf = time.time()
        times_exact.append(tf - ts)

        diff = probs_fermion-probs_exact
        diffs.append(np.abs(diff))
        # See if two simulations agree:
        print('Fermion_prob, ED_prob, difference', (probs_fermion, probs_exact, diff))

    avg_time_fermion = sum(times_fermion)/reps
    std_time_fermion = np.std(times_fermion)

    avg_time_exact = sum(times_exact)/reps
    std_time_exact = np.std(times_exact)


    avg_time_at_N_fermion.append(avg_time_fermion)
    avg_time_at_N_exact.append(avg_time_exact)
    
    std_time_at_N_fermion.append(std_time_fermion)
    std_time_at_N_exact.append(std_time_exact)
    
    avg_abs_diff_prob_at_N.append(np.mean(diffs))
    std_abs_diff_prob_at_N.append(np.std(diffs))

    avg_time_at_N_fermion = np.array(avg_time_at_N_fermion)

# Plots
avg_time_at_N_exact = np.array(avg_time_at_N_exact)

std_time_at_N_fermion = np.array(std_time_at_N_fermion)
std_time_at_N_exact = np.array(std_time_at_N_exact)
    
avg_abs_diff_prob_at_N = np.array(avg_abs_diff_prob_at_N)
std_abs_diff_prob_at_N = np.array(std_abs_diff_prob_at_N)
    
print('avg_time_at_N_fermion', avg_time_at_N_fermion)
print('avg_time_at_N_exact', avg_time_at_N_exact)

plt.figure(figsize=(10, 6))
# plt.plot(Ns, avg_time_at_N_exact+std_time_exact, '^')
# plt.plot(Ns, avg_time_at_N_exact-std_time_exact, 'v')
plt.fill_between(Ns, avg_time_at_N_exact-std_time_at_N_exact, \
                     avg_time_at_N_exact+std_time_at_N_exact, color='mistyrose')
plt.fill_between(Ns, avg_time_at_N_fermion-std_time_at_N_fermion, \
                     avg_time_at_N_fermion+std_time_at_N_fermion, color='lightblue')

plt.plot(Ns, avg_time_at_N_fermion, "^-")
plt.plot(Ns, avg_time_at_N_exact, "v-")
plt.legend(['avg_time_at_N_exact', 'avg_time_at_N_fermion'], fontsize=20)
# plt.title('Runtime vs N-qubit sizes')
# plt.yscale('log')
plt.ylabel('runtime', fontsize=20)
plt.xlabel('number of qubits', fontsize=20)
plt.xticks(fontsize=20)
plt.yticks(fontsize=20)

plt.figure(figsize=(10, 6))
plt.fill_between(Ns, avg_abs_diff_prob_at_N-std_abs_diff_prob_at_N, \
                 avg_abs_diff_prob_at_N+std_abs_diff_prob_at_N, color='lightblue')
plt.plot(Ns, avg_abs_diff_prob_at_N, '*-')
plt.xticks(fontsize=20)
plt.yticks(fontsize=20)
# plt.title('Avg of absolute val of prediction difference')
plt.xlabel('number of qubits', fontsize=20)
plt.ylabel('|P_f - P_e|', fontsize=20)
\end{lstlisting}

\begin{figure}[h!]
  \centering
  \begin{minipage}[b]{0.47\textwidth}
    \includegraphics[width=\textwidth]{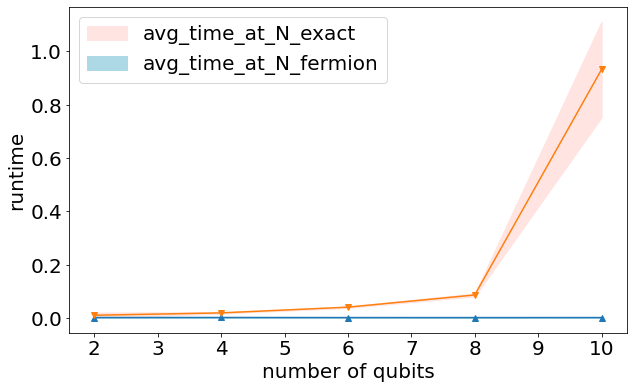}
    \caption{Average runtime versus number of qubits ($N$) in system. Shaded area is one standard above and below the mean for $20$ repeated runs at each $N$. We see an exponential increase in computing time for exact diagonalization.}
  \end{minipage}
  \hfill
  \begin{minipage}[b]{0.50\textwidth}
    \includegraphics[width=\textwidth]{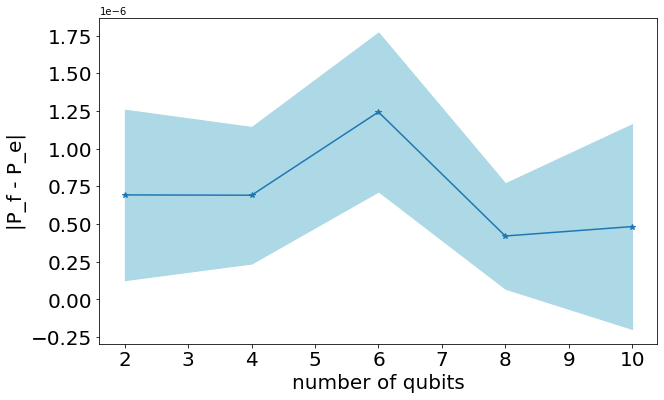}
    \caption{The average of density evaluation difference between exact diagonalization and NFNet, i.e., $|p_{fermion}(y=x|x)-p_{exact}(y=x|x)|$ at each N. The numerical differences are in the order of 1e-6.}
  \end{minipage}
\end{figure}

\subsection{Multi-layer Circuit and 512+ qubit measurement computation} \label{sec:multicircuit-512}
The power of NFNet really comes down to the fact that it can efficiently model a series of quantum gates, each corresponds to a continuous Hamiltonian time evolution, forming a multi-layer parameterized circuit. The network takes an input state denoted by a bitstring $\x$, and computes the probability density value of observing any bitstring $\y$ in polynomial time of the number of qubits. 

More formally, a ``NFNet.Network'' object is a parameterized computing model $f_{\theta}(y;x)$, which takes in two bitstrings $y$ and $x$ and outputs the probability density $p(y|x)$. This model is multi-purpose, and can be used as a general quantum ansatz, a quantum generative/classification model, a classical quantum simulation platform for quantum systems, etc.

Here we demonstrate how to build a multi-layer parameterized quantum circuit with NFNet classes and functions.

We define one circuit gate as $H_g$, and each gate corresponds to a continuous evolution $e^{-i\mathcal{H}t}$. A parameterized is a series of such evolutions:
\begin{equation}
    U_{total} = U_L U_{L-1}...U_{2} U_1,
\end{equation}
where $U_l$ is the $l$-th "layer" defined as:
\begin{equation}
    U_l = e^{-i\mathcal{H}_{l}t}
\end{equation}

where $\mathcal{H}_{l}$ concerns interactionos of two Fermionic modes $i<j$, not necessarily nearest-neighbor. We first compute the measurement probabilities of all possible outcome bitstrings $\y$ with NFNet, and then compare these probabilities with those computed by exact diagonalization. The circuit we will use is a full nearest-neighbor scheme repeated 3 times.


\begin{lstlisting}[language=Python]
# Load NFNet modules:
from Utils_torch_version import Network, get_nn_pairs, binary_basis, unpacknbits, initialize_sparse
import numpy as np
import matplotlib.pyplot as plt
import scipy.linalg
import time
import torch # PyTorch is required for this demo
import qiskit # Qiskit is required for this demo
import math

# Define the model configurations
N = 10
# Connectivity (i,j) of each "layer" where i<j.
conn_list = [[np.min(p), np.max(p)] for p in get_nn_pairs((N,))]*3 # zero indexed and should not be periodic (not a closed circle)
L = len(conn_list)

# Define the circuit, calling the FermiNet "Network" class
circuit = Network(conn_list, N)

# Define an input state, in this case, a simple 1010...1010 basis state
# Note: x should follow the shape (batch_size, 2^N), even for a batch_size=1
x = torch.tensor([[1,0]*int(N/2)]) if N%2==0 else torch.tensor([[1,0]*int(N/2)+[1]])

# When the circuit is initialized, the parameters are randomly initiailized by pi * Uniform([0,1]):
for param in circuit.parameters():
    print('Network parameters', param)

# Hereby we take an extra step to show you how to manually set these parameters:
params_m = torch.tensor(math.pi) * torch.rand((L, 4)) # Each layer has 4 rotations
print('params_m',params_m)

circuit.manual_set_params(params_m)

# Compute the probabilities of observing all possible bitstrings
# FermiNet supports batch processing too!

probs = torch.zeros(2**N, dtype=torch.cfloat) # To collect the probability densitives evaluations
basis_m_n = torch.tensor(binary_basis(geometry=(N,)))

# This demonstrates how to process P(y|x) by batches: P(y_batch|x_batch) results in a batch of probabilities
batch_size = 16
n_batches = len(probs)//batch_size if len(probs)%batch_size == 0 else len(probs)//batch_size+1

ts = time.time()
for i in range(n_batches):
    y_batch = basis_m_n[batch_size*i : batch_size*(i+1)] # shape is (batch_size, 2^N)
    x_batch = x.repeat(y_batch.shape[0], 1) # shape is (batch_size, 2^N)
    
    # In this example, we measure all the qubits at the end, so the mask is 1111..111
    sub_mask_batch = (torch.tensor([ [1]*N ])).repeat(y_batch.shape[0], 1) # shape is (batch_size, 2^N)
    # call circuit.forward_partial_observation() to get the P(y|x) values
    probs_batch = circuit.forward_partial_observation(y_batch, x_batch, sub_mask_batch)
    # Now put these probs_batch values into correct positions
    probs[batch_size*i : batch_size*(i+1)] = probs_batch
print('probs', probs)
# The returned probs are torch tensors, so detach the gradients and convert to numpy:
probs = probs.detach().numpy().real

tf = time.time()
print('time lapsed: ', str(tf-ts))
\end{lstlisting}
Then we repeat the probability computations with exact diagonalization:
\begin{lstlisting}[language=Python]
x_string = '10'*int(N/2)+'1' if N%2==1 else '10'*int(N/2)
init_state_numpy = qiskit.quantum_info.Statevector.from_label(x_string).data

ts = time.time()
# Initialize the evolution matrix
exp_iH_exact = np.eye(2**N)
# Iterate each layer and get the composition via matrix multiplication
# Use the same params_m rotations defined earliers
for l in range(L):
    conn = conn_list[l]
    a, b, c, d = params_m.detach().numpy()[l]
    H_exact = initialize_sparse(N, conn[0], conn[1], a, b, c, d)
    exp_iH_exact = (scipy.linalg.expm(-1.0j*H_exact))@exp_iH_exact 

state_exact = np.matmul(exp_iH_exact, init_state_numpy[:,None])
probs_exact = (np.abs(state_exact)**2).squeeze()

tf = time.time()
\end{lstlisting}
Then we compare the probabilities. The difference is measured via the total variation: $TV=\sum_y |p(y|x)-q(y|x)|$, where $p(y|x)$ is the Fermion-simulated density and $q(y|x)$ is exact-diagonalizaion-calculated density. In this specific example, the TV value is 1.011e-12.
\begin{lstlisting}[language=Python]
# Plot two against each other
fig, ax = plt.subplots(figsize=(20, 10))
xs = np.arange(2**N)
ax.plot(xs, probs, color='lightcoral')
ax.plot(xs, probs_exact, '^', color='lightblue')

ax.set_xlabel('bit strings in base 10', fontsize=15)
ax.set_ylabel('probability', fontsize=15)

ax.legend(['Fermion Simulation', 'Exact Diagonalization'], fontsize=15)
tv = (np.abs(probs-probs_exact)**2).sum()
ax.set_title('Total variation: '+str(tv))
\end{lstlisting}
\begin{figure}
    \centering
    \includegraphics[scale=0.33]{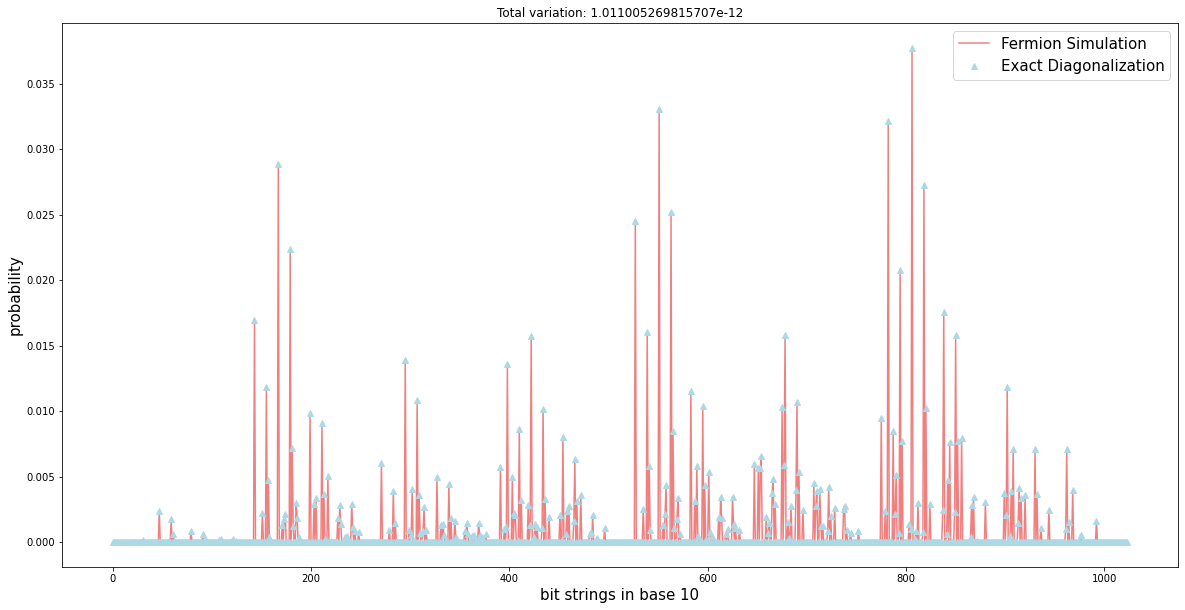}
    \caption{Comparison between the Fermion simulated distribution and exact-diagonalization-computed distribution. The TV value is 1.011e-12.}
    \label{fig:my_label}
\end{figure}

With the same circuit architecture, we now look at the runtime scaling of the measurement probability calculation using NFNet at the mega scale. Namely, for qubit sizes ranging from $100$ to $1000$,  we compute the measurement probability $p_{\theta}(y|x)$ and record the runtime.
\begin{lstlisting}[language=Python]
from Utils_torch_version import Network, get_nn_pairs, binary_basis, unpacknbits, initialize_sparse

import numpy as np
import matplotlib.pyplot as plt
#from scipy import sparse
#import scipy
import scipy.linalg
import time
import torch # PyTorch is required for this demo
import math
import qiskit # Qiskit is required for this demo

Ns = [100*(i+1) for i in range(10)]
reps = 1
avg_time_at_N_fermion = []
std_time_at_N_fermion = []
probs_fermion_list = []
for N in Ns:
    print('start N: ', N)
    conn_list = [[np.min(p), np.max(p)] for p in get_nn_pairs(geometry=(N,))]*1 # zero indexed
    times_fermion = []
    times_exact = []
    diffs = [] # to collect the differences between two simulation methods
    for rep in range(reps):
        L = len(conn_list) # Number of layers
        # initiliaze the evolution
        evolution = Network(conn_list, N)
        
        # The 1010101... basis state as input
        x = torch.tensor([ [1,0]*int(N/2) ]) if N%2==0 else torch.tensor([[1,0]*int(N/2)+[1]])  
    
        ts = time.time()
        
        y = x # Calculate the probability of observing 101010... Can choose other bit strings too
        sub_mask = torch.tensor([ [1]*N ])# Eg. 11111 means measuring all qubits, 11000, means measurin the first two qubits
        
        # Since we aren't doing optimization, simply record the numerical value, detach gradients.
        probs_fermion = \
            evolution.forward_partial_observation(y, x, sub_mask).detach().numpy()
        print('debug probs_fermion', probs_fermion)
        probs_fermion_list.append(probs_fermion)
        # Record the Fermion simulation's runtime
        tf = time.time()
        times_fermion.append(tf - ts)

        ts = time.time() # Reset timer
        
    avg_time_fermion = sum(times_fermion)/reps
    std_time_fermion = np.std(times_fermion)
    avg_time_at_N_fermion.append(avg_time_fermion)
    std_time_at_N_fermion.append(std_time_fermion)

avg_time_at_N_fermion = np.array(avg_time_at_N_fermion)
std_time_at_N_fermion = np.array(std_time_at_N_fermion)
    
plt.figure(figsize=(10, 6))
plt.fill_between(Ns, avg_time_at_N_fermion-std_time_at_N_fermion, \
                     avg_time_at_N_fermion+std_time_at_N_fermion, color='lightblue')
plt.plot(Ns, avg_time_at_N_fermion, "^-", color='lightcoral')
plt.yscale('log')
plt.ylabel('runtime', fontsize=20)
plt.xlabel('number of qubits', fontsize=20)
plt.xticks(fontsize=20)
plt.yticks(fontsize=20)
\end{lstlisting}
\begin{figure}[h!]
    \centering
    \includegraphics[scale=0.5]{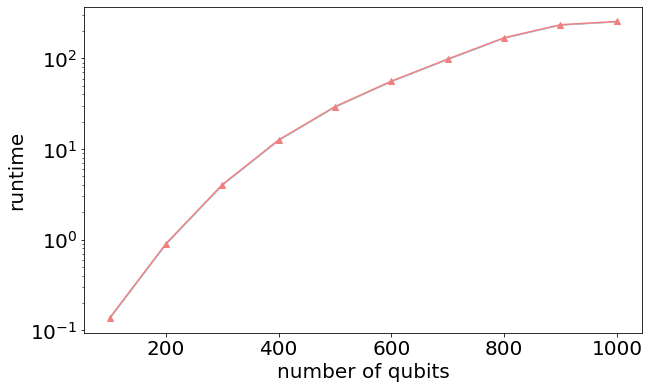}
    \caption{Runtime (in seconds) scaling of measurement probability computation at different system sizes. The plot indicates polynomial scaling.}
    \label{fig:my_label}
\end{figure}

\subsection{196-qubit Single Pattern Recognition} \label{subsec:196-mnist}
In this example, we test the circuit's trainability at the 196-qubit level. Namely, we optimize the circuit parameters such that when we measure the circuit, we get a specific bit string pattern with high probability $\approx 1$. Namely, given a target pattern $\y$, we optimize the circuit parameters to maximize $p_{\theta}(y|x)$, given an arbitrary input state $|x\rangle$ with bit string $\x$. The target bit string is chosen to be a binary pattern resembling an MNIST digit ``9'', calling the inbuilt NFNet function ``NFNet.mnist\textunderscore single\textunderscore pdf''. In this proof-of-concept example, we set $\x=\y$ as the output probabilities are extremely sparse, and searching for an input $\x$ with nonzero output density is exponentially prohibitive. Therefore, at optimum, the unitary transformation is the identity.

\begin{lstlisting}[language=Python]
from Utils_torch_version import Network, get_nn_pairs, unpacknbits, initialize_sparse, mnist_single_pdf
import numpy as np
import matplotlib.pyplot as plt
import time
import torch # PyTorch is required for this demo
import qiskit # Qiskit is required for this demo
import math
import seaborn as sns
import torch
import matplotlib.pyplot as plt
import qiskit
import numpy as np

q_data, data_digit = mnist_single_pdf(9, 17, 14)
q_data_img = (q_data>0.001).int()
q_data = q_data_img.flatten().unsqueeze(0)

n_fermions = 68
factor = 1
N = 14**2
conn_list = [[np.min(p), np.max(p)] for p in get_nn_pairs((N,))]*1
print('conn_list', conn_list)
L = len(conn_list) # Number of layers

# Initialize the circuit
circuit = Network(conn_list, N)
print('circuit.parameters()', circuit.parameters())

# Initialize the circuit
circuit = Network(conn_list, N)

beta1 = 0.5
lr_G = 1e-1
optimizerG = torch.optim.Adam(circuit.parameters(), lr=lr_G, betas=(beta1, 0.999))

prob_list = []
for itr in range(20):
    circuit.zero_grad() # clear the parameter gradients
    y_batch = q_data # Data pattern
    x_batch = q_data # In this case, set x=y 
    sub_mask_batch = (torch.tensor([ [1]*(N) ])).repeat(y_batch.shape[0], 1)
    ts = time.time()
    probs_batch = -circuit.forward_partial_observation(y_batch, x_batch, sub_mask_batch).real
    print('iter, obj', (itr, probs_batch))
    
    probs_batch.backward() # Get gradients
    optimizerG.step() # Update
    tf = time.time()
    print('total time', tf-ts)
    prob_list.append(probs_batch)
prob_list = torch.tensor(prob_list)
plt.figure(dpi=200)
plt.plot(-prob_list, color='lightcoral')
plt.xlabel('iteration')
plt.ylabel('prob')
plt.title('prob of outputing target pattern')
\end{lstlisting}

\begin{figure}[!h] 
  \centering
  \begin{minipage}[b]{0.3\textwidth}
    \includegraphics[width=\textwidth]{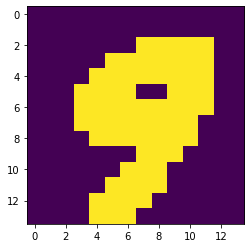}
    \label{recognition}
  \end{minipage}
  \hfill
  \begin{minipage}[b]{0.45\textwidth}
    \includegraphics[width=\textwidth]{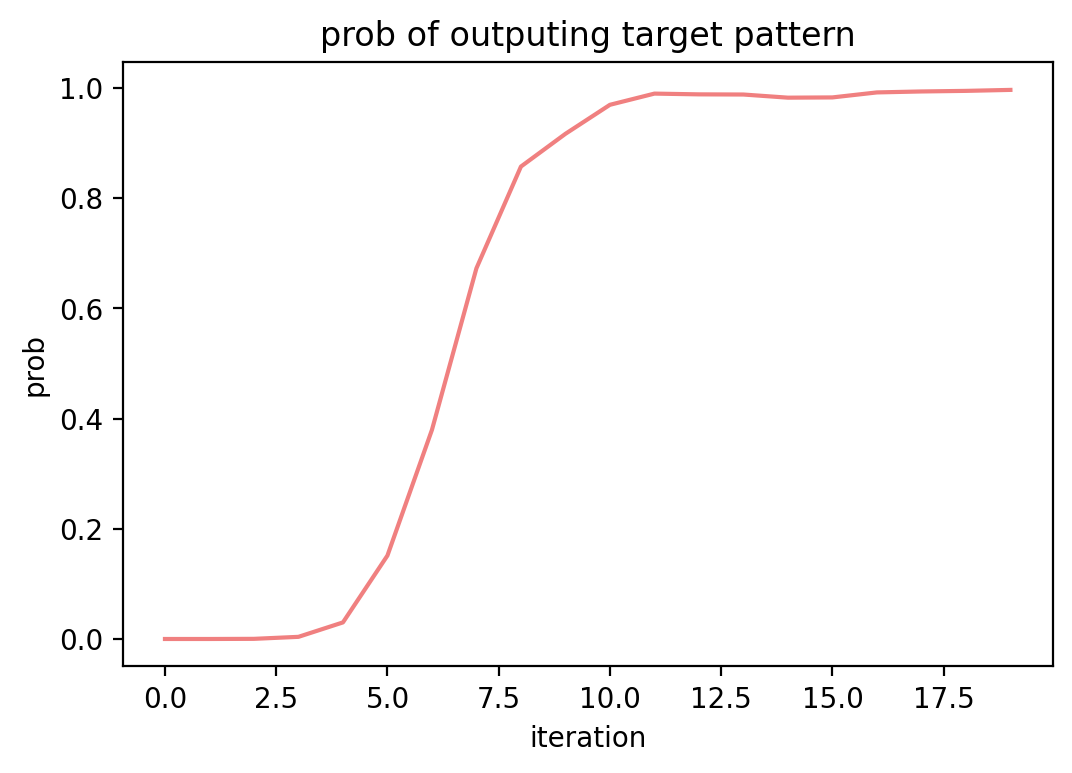}
  \end{minipage}
  \caption{Left: Target pattern. Right: probability of observing this pattern vs iterations.}
\end{figure}

\subsection{Single MNIST generation with Quantum Born Machine} \label{subsec:born}
The quantum circuit Born Machine (QCBM) is a quantum generative model[Ref]. QCBM utilizes a multi-layer parameterized quantum circuit (MPQC) to evolve the initial/input quantum state $|\psi_0\rangle$ to some target state via unitary gates: $\left|\psi_{\theta}\right\rangle=U_{\theta}|\psi_0\rangle$, where $\theta$ are the parameters of the MPQ. One measures the outputs state in the computational basis to produce a classical sample (bit string) $x \sim p_{\boldsymbol{\theta}}(x)=\left|\left\langle x | \psi_{\theta}\right\rangle\right|^{2}$.

Excitingly, the output probability densities of a general quantum circuit cannot be efficiently simulated by classical means, the QCBM is among the several proposals to show quantum supremacy [ref]. QCBM is conventionally trained by minimizing the maximum mean discrepancy (MMD) loss using a Gaussian Kernel:
\begin{equation}
    \begin{aligned}
\mathcal{L}&=\underset{x \sim p_{\theta}, y \sim p_{\theta}}{\mathbb{E}}[K(x, y)]-2 \underset{x \sim p_{\theta}, y \sim p^*}{\mathbb{E}}[K(x, y)]+\underset{x \sim p^*, y \sim p^*}{\mathbb{E}}[K(x, y)].
\end{aligned}
\end{equation}

In this tutorial, we use a simple shallow parameterized quantum circuit similar in Section \ref{sec:multicircuit}. The target distribution is a normalized gray-scale MNIST digit. Note that the circuit expressiveness can be improved by other connection schemes which we do not particularly study in this tutorial.
\begin{lstlisting}[language=Python]
# Load NFNet modules:
from Utils_torch_version import Network, get_nn_pairs, binary_basis, unpacknbits, initialize_sparse, \
    mnist_single_pdf, mix_rbf_kernel, kernel_expect
import numpy as np
import matplotlib.pyplot as plt
import scipy.linalg
import time
import torch # PyTorch is required for this demo
import qiskit # Qiskit is required for this demo
import math
import seaborn as sns

pdf_data_img,_ = mnist_single_pdf(0, 9, (5,5)) # Use NFNet inbuilt util function to load a mnist normalized grayscale image
plt.imshow(pdf_data_img)
# Then, flatten out the distribution into 1d tensor and create the target data distribution:
pdf_data = torch.zeros(2**5, dtype=torch.float)
pdf_data[:25] = torch.tensor(pdf_data_img.flatten(), dtype=torch.float)
\end{lstlisting}
We then set up the Born machine learning problem, where we use the ``NFNet.Network'' object as the quantum ansatz. Further more, as ``NFNet'' support PyTorch-based auto-differentiation, directly use the PyTorch Adam optimizer to update the circuit parameters.
\begin{lstlisting}[language=Python]
n_fermions = 5
factor = 10
N = factor * n_fermions
basis_m_n = binary_basis(geometry=(n_fermions,)) # The basis in the probability space

conn_list = [[i, N-1] for i in range(N-1)]*10
print('conn_list', conn_list)
L = len(conn_list) # Number of layers

# 101010...10
x_input = torch.tensor([[1]*(N//factor)+[0]*(N-N//factor)]) if N%2==0 else torch.tensor([[1,0]*int(N/2)+[1]])

# Initialize the circuit
circuit = Network(conn_list, N)
print('circuit.parameters()', circuit.parameters())

# Initialize the circuit
circuit = Network(conn_list, N)

beta1 = 0.5
lr_G = 1e-2
optimizerG = torch.optim.Adam(circuit.parameters(), lr=lr_G, betas=(beta1, 0.999)) #The parameters are the th

# MMD loss by tracking the full probability space [0.5, 0.1, 0.2, 0.25,4,10]
K = torch.tensor(mix_rbf_kernel(basis_m_n, basis_m_n, sigma_list=[2e-6]), dtype=torch.float)
def exact_mmd(pdf_data, pdf_model): #input are tensors
    p_diff = pdf_data-pdf_model # Although this puts a constant term to the loss value, it is easier to code this way
    return kernel_expect(K, p_diff, p_diff)

n_space = basis_m_n.shape[0] # The number of total probability patterns
batchsize = 1024
num_batches = n_space//batchsize if n_space%batchsize==0 else n_space//batchsize+1

for itr in range(500): 
    pdf_model = torch.zeros(basis_m_n.shape[0]) # To keep full pdf for exact calculation of the MMD loss
    circuit.zero_grad() # clear the parameter gradients
    obj = torch.tensor(0.0)
    # Loop through all possible patterns in the total probability space
    for i in range(num_batches):
        y_batch = basis_m_n[i*batchsize:(i+1)*batchsize]
        x_batch = x_input.repeat_interleave(y_batch.shape[0], axis=0)
        sub_mask_batch = (torch.tensor([ [1]*(N//factor)+[0]*(N-N//factor) ])).repeat(y_batch.shape[0], 1) # Measure the first half of the qubits

        probs_batch = circuit.forward_partial_observation(y_batch, x_batch, sub_mask_batch)
        # Only keep the real part, as all information goes to the real part
        probs_batch = probs_batch.real
        #print('probs_batch', probs_batch)
        pdf_model[i*batchsize:(i+1)*batchsize] = probs_batch # Gradients carry through, as this directly goes into the loss calculation

    obj = exact_mmd(pdf_data, pdf_model)
    obj.backward() # Calculate parameter gradients
    optimizerG.step() # Update parameters
    print('iter, obj', (itr, obj))

    if itr % 20 == 0:
        plt.figure()
        plt.imshow(pdf_model[:25].detach().numpy().reshape(5,5))
        plt.savefig('img/img_iter_'+str(itr)+'.png')
        plt.imsave('img/img_iter_'+str(itr)+'.png',pdf_model[:25].detach().numpy().reshape(5,5))
        plt.show()
\end{lstlisting}
We present the learning result as follows. 

\begin{figure}[h!]
  \centering
  \begin{minipage}[b]{0.2\textwidth}
    \includegraphics[width=\textwidth]{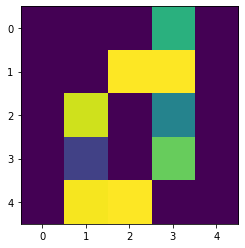}
  \end{minipage}
  \hfill
  \begin{minipage}[b]{0.2\textwidth}
    \includegraphics[width=\textwidth]{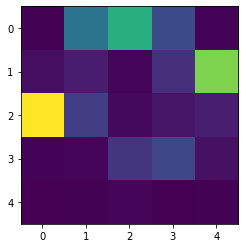}
  \end{minipage}
  \hfill
  \begin{minipage}[b]{0.2\textwidth}
    \includegraphics[width=\textwidth]{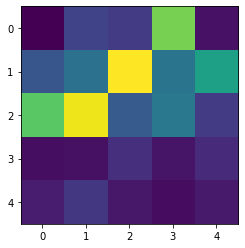}
  \end{minipage}
  \hfill
  \begin{minipage}[b]{0.2\textwidth}
    \includegraphics[width=\textwidth]{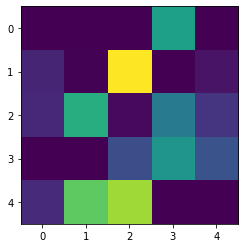}
  \end{minipage}
  \caption{Left to right: Target Distribution; randomly initialized output distribution; learnt circuit output distribution at iteration $10$; learnt circuit output distribution at iteration $50$.}
\end{figure}
\begin{figure}[h!]
    \centering
    \includegraphics[scale=0.3]{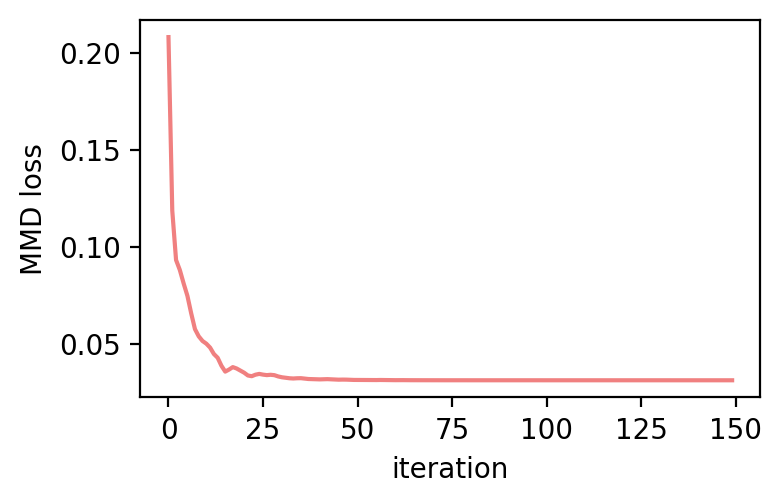}
    \caption{The MMD loss versus iteration number.}
    \label{fig:my_label}
\end{figure}

\subsection{Weighted-edge Maxcut} \label{subsec:maxcut}
The Maxcut objective can be expressed by a classical Hamiltonian (with measurements only in Z basis). Such Hamiltonians nicely fit in the free Fermion formulation. In this example, we showcase how to use NFNet to solve a MaxCut problem, whose Hamiltonian is:
\begin{equation}
    H_{maxcut}=\frac{1}{2} \sum_{i<j} c_{ij}(I - Z_{i}Z_{j}).
\end{equation}

In this example, we use a simple quantum ansatz modeled by a multi-layer parametrized quantum circuit with related NFNet classes and functions.

We first set up the graph and get the exact optimal solution via a classical algorithm provided by Qiskit.
\begin{lstlisting}[language=Python]
    
from Utils_torch_version import Network, get_nn_pairs, binary_basis, initialize_sparse

import numpy as np
import matplotlib.pyplot as plt
import scipy.linalg
import time
import torch # PyTorch is required for this demo
import qiskit # Qiskit is required for this demo
import math
import seaborn as sns
# Other packages required in this example:
import torch
import matplotlib.pyplot as plt
import qiskit
import numpy as np
import networkx as nx
from qiskit_optimization.applications import Maxcut, Tsp
from qiskit_optimization.algorithms import MinimumEigenOptimizer
from qiskit.algorithms import VQE, NumPyMinimumEigensolver

n_fermions = 12 # 12 nodes in graph
n_show_qubits = n_fermions

connections = [(5,0), (0, 1), (1, 2), (2, 3), (3, 4)] # NOT NECESSARY
#======================================Create p_data==========================================
hdim = 2**n_show_qubits
G = nx.Graph()
G.add_nodes_from([i for i in range(n_show_qubits)])

edges = []
for i in range(n_show_qubits):
    for j in range(n_show_qubits):
        if j > i:
            edges.append((i,j, round(np.random.rand(), 2) ))
G.add_weighted_edges_from(edges)

colors = ["r" for node in G.nodes()]
pos = nx.spring_layout(G)


def draw_graph(G, colors, pos):
    plt.figure(3,figsize=(12,12)) 
    default_axes = plt.axes(frameon=True)
    nx.draw_networkx(G, node_color=colors, node_size=600, alpha=0.8, ax=default_axes, pos=pos)
    edge_labels = nx.get_edge_attributes(G, "weight")
    nx.draw_networkx_edge_labels(G, pos=pos, edge_labels=edge_labels)


draw_graph(G, colors, pos)
# Computing the weight matrix from the random graph
W = np.zeros([n_show_qubits, n_show_qubits])
for i in range(n_show_qubits):
    for j in range(n_show_qubits):
        temp = G.get_edge_data(i, j, default=0)
        if temp != 0:
            W[i, j] = temp["weight"]

# Solve for the exact solution with classical method
max_cut = Maxcut(W)
qp = max_cut.to_quadratic_program()
print(qp.prettyprint())

# solving Quadratic Program using exact classical eigensolver
exact = MinimumEigenOptimizer(NumPyMinimumEigensolver())
result = exact.solve(qp)
print(result.prettyprint())
x = result.x
colors = ["r" if x[i] == 0 else "c" for i in range(n_fermions)]
draw_graph(G, colors, pos)
\end{lstlisting}
We then set up the circuit ansatz and define the objective function:
\begin{lstlisting}[language=Python]
factor = 2
N = factor * n_show_qubits # Total qubits
basis_m_n = binary_basis(geometry=(n_show_qubits,)) # The basis in the probability space

conn_list = [ [np.min(p), np.max(p)] for p in get_nn_pairs(geometry=(N,))]*10

print('conn_list', conn_list)
L = len(conn_list) # Number of layers

# 101010...10
x_input = torch.tensor([[1,0]*int(N/2)]) if N%2==0 else torch.tensor([[1,0]*int(N/2)+[1]])

# Initialize the circuit
circuit = Network(conn_list, N)
print('circuit.parameters()', circuit.parameters())

# Initialize the circuit
circuit = Network(conn_list, N)

beta1 = 0.9
lr_G = 1e-1
optimizerG = torch.optim.Adam(circuit.parameters(), lr=lr_G, betas=(beta1, 0.999)) #The parameters are the th

# pre-conpute the weight matrix W (symmetric) to save time
def maxcut_obj(x, G, W):
    obj = 0
    for i, j in G.edges():
        if x[i] != x[j]:
            obj -= W[i][j]
            
    return obj

# pl is the probability list (normalized or not), either exact or with sampling error
def compute_expectation(counts, x_strings, G, W):
    avg = 0
    sum_count = 0
    for i in range(len(counts)):
        bitstring = x_strings[i]
        count = counts[i]
        obj = maxcut_obj(bitstring, G, W)
        avg += obj * count
        sum_count += count
        
    return avg/sum_count
\end{lstlisting}
Then the optimization process is similar as before:
\begin{lstlisting}[language=Python]
n_space = basis_m_n.shape[0] # The number of total probability patterns
batchsize = 1024*10
num_batches = n_space//batchsize if n_space%batchsize==0 else n_space//batchsize+1
obj_list = []
for itr in range(20): # At each iteration, measure the kl divergence and update
    #probs_theta_valid_samples = torch.zeros(bas_samples.shape[0]) # to collect the model probabilities at valid patterns
    pdf_model = torch.zeros(basis_m_n.shape[0]) # To keep full pdf for exact calculation of the MMD loss

    ts = time.time()
    circuit.zero_grad() # clear the parameter gradients
    obj = torch.tensor(0.0)
    # Loop through all possible patterns in the total probability space
    for i in range(num_batches):
        y_batch = basis_m_n[i*batchsize:(i+1)*batchsize]
        x_batch = x_input.repeat_interleave(y_batch.shape[0], axis=0)
        sub_mask_batch = (torch.tensor([ [1]*(N//factor)+[0]*(N-N//factor) ])).repeat(y_batch.shape[0], 1) # Measure the first half of the qubits

        probs_batch = circuit.forward_partial_observation(y_batch, x_batch, sub_mask_batch)
        # probs_batch = torch.conj(amps_batch)*amps_batch # this is still torch.cfloat
        # Only keep the real part, as all information goes to the real part
        probs_batch = probs_batch.real
        #print('probs_batch', probs_batch)
        pdf_model[i*batchsize:(i+1)*batchsize] = probs_batch # Keep track of the gradient, as this directly goes into the loss calculation

    tf = time.time()
    # in the case of calculating the exact MMD loss, which cannot be written as a sum
    obj = compute_expectation(pdf_model, basis_m_n, G, W)
    print('iter, obj, time', (itr, obj, tf-ts))
    obj.backward()
    optimizerG.step()
    
    obj_list.append(obj)
\end{lstlisting}

\begin{figure}[h!]
  \centering
  \begin{minipage}{0.32\textwidth}
    \includegraphics[width=\textwidth]{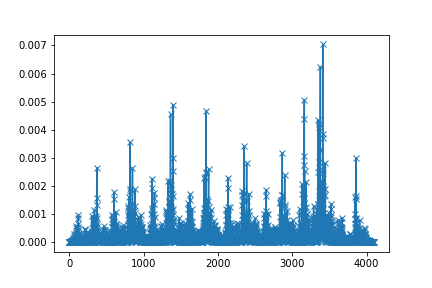}
  \end{minipage}
  \hfill
  \begin{minipage}{0.32\textwidth}
    \includegraphics[width=\textwidth]{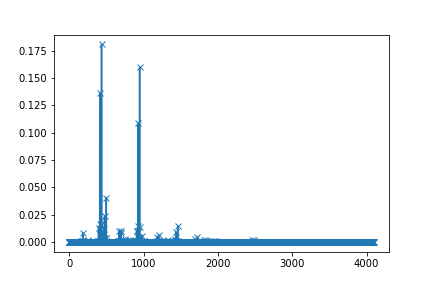}
  \end{minipage}
  \hfill
  \begin{minipage}{0.32\textwidth}
    \includegraphics[width=\textwidth]{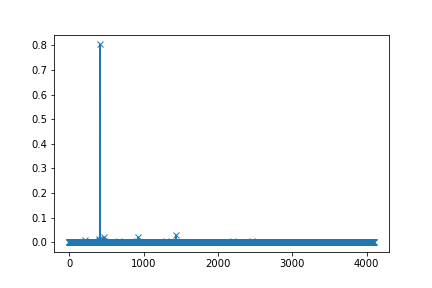}
  \end{minipage}
  \caption{Probability distribution during training. Left to right: random initialization; iteration 11; iteration 20. At iteration 11, a transition happens that shifts from one peak to the other.}
\end{figure}

\begin{figure}[h!]
    \centering
    \includegraphics[scale=0.2]{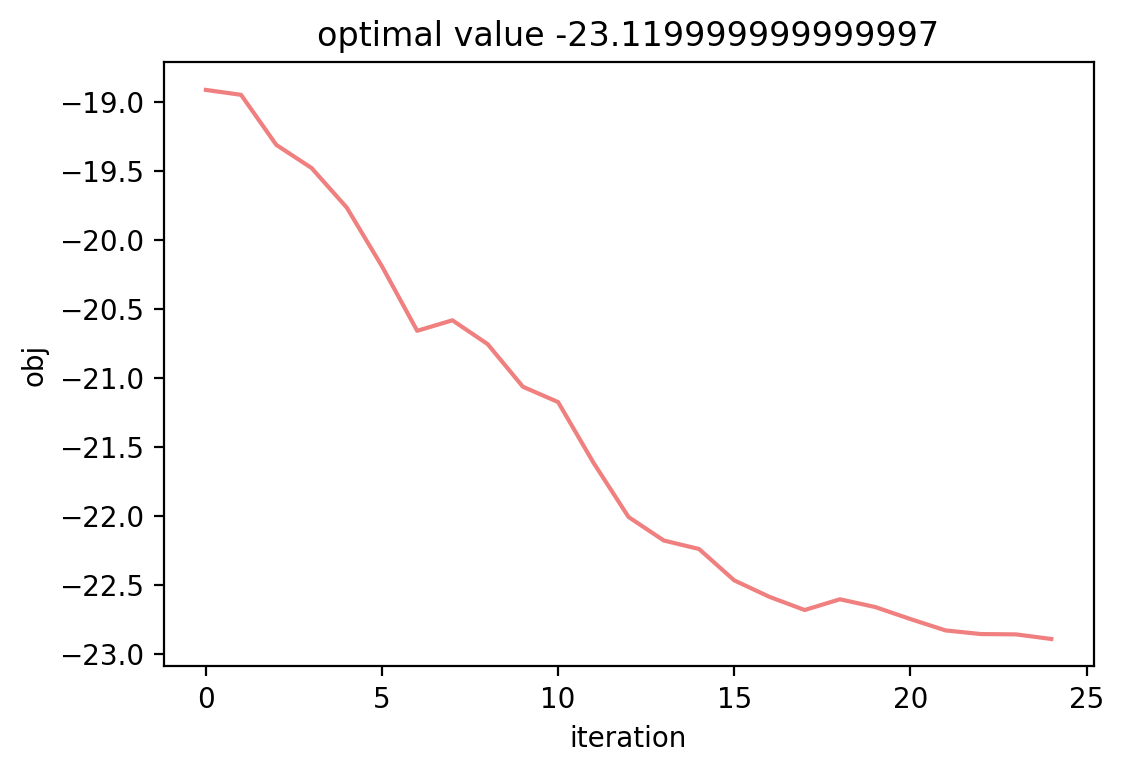}
    \caption{Maxcut objective function vs iteration number}
    \label{fig:my_label}
\end{figure}
And the found (most likely to output) graph is indeed the optimum (as computed via classical algorithm) Fig. \ref{fig:maxcut}.
\begin{figure}[h!]
    \centering
    \includegraphics[scale=0.5]{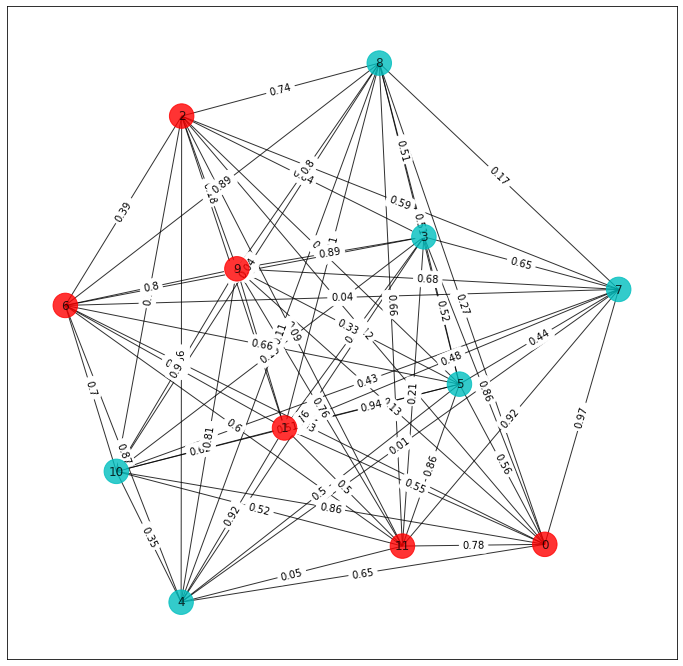}
    \caption{Learnt partition after 25 iterations.}
    \label{fig:maxcut}
\end{figure}

\section{Conclusion and Future Extensions} \label{sec:conclusion}
We present the mathematical foundation and design principles of NFNet (Non-interacting Fermion Net) \cite{NFNet-github}, an open-source software platform for large-scale classical simulation of continuously controlled quantum systems, based on an antisymmetric matrix decomposition that relates a special set of unitary gates to non-interacting Fermions. NFNet computes the probability (P(y|x)) of full or partial measurement (y) on a N-qubit quantum circuit in polynomial time of N. We proved that the strong simulability of NFNet goes beyond the $2$-local quantum circuit setting and extends to dense Hamiltonians written as a sum of arbitrary two-mode Fermionic interactions. To enhance NFNet’s performance, we engineered parallel processing of measurement simulations and continuous auto-differentiation of circuit/system parameters, which makes NFNet not only a classical quantum simulator, but also a quantum-mechanics-inspired machine learning framework. To benefit a broader research community, we presented NFNet tutorials on weighted-edge Maxcut, 196-qubit pattern recognition of MNIST hand-written digits, and measurement benchmarks with 512+ qubits with the interactive IPython environment runnabled even on a personal laptop.

Classically, NFNet is a computing framework for normalized outputs on a discrete support, as the inherent probabilistic nature of quantum system naturally normalizes the output probability densities. Delving into the basic arithmetic operators, we develop intuitive connections between NFNet and RNN, power iteration methods, and suggest its potential as a discrete-space normalizing flow model. Future directions include improving the code structure of NFNet, enriching its functionalities, and studying its furture connections to known computational models.

\newpage
\bibliographystyle{unsrt}
\bibliography{references}
\end{document}